\documentclass{llncs}
\usepackage[dvipsnames]{xcolor}
\usepackage{amssymb}
\usepackage{amsmath}
\usepackage{hyperref}
\usepackage{tikz}
\usepackage{subcaption}
\usepackage{caption}
\usepackage[capitalize]{cleveref}
\usepackage{paralist}
\usepackage{booktabs}
\usepackage{thmtools}
\usepackage{wrapfig}
\usepackage{xspace}
\usepackage{colortbl}
\usepackage{stmaryrd}
\usepackage[normalem]{ulem}
\usepackage[noend,ruled]{algorithm2e}

\usepackage{txfonts}

\usetikzlibrary{arrows,arrows.meta,automata,fit,backgrounds,bending,calc,positioning}

\newcommand{\ol}[1]{\textcolor{blue}{\ifmmode \text{[OL: #1]}\else [OL: #1] \fi}}
\newcommand{\vh}[1]{\textcolor{green!50!black}{\ifmmode \text{[VH: #1]}\else [VH: #1] \fi}}
\newcommand{\mh}[1]{\textcolor{purple}{\ifmmode \text{[MH: #1]}\else [MH: #1] \fi}}
\newcommand{\ph}[1]{\textcolor{red}{\ifmmode \text{[PH: #1]}\else [PH: #1] \fi}}
\newcommand{\lh}[1]{\textcolor{magenta}{\ifmmode \text{[LH: #1]}\else [LH: #1] \fi}}

\newcommand{\integers}[0]{\mathbb{Z}}
\newcommand{\bin}[0]{\mathbb{B}}
\newcommand{\semof}[1]{\llbracket#1\rrbracket}
\newcommand{\modencof}[1]{\langle#1\rangle}
\newcommand{\decof}[1]{\mathsf{dec}(#1)}
\newcommand{\wordof}[1]{\mathcal{L}(#1)}

\newcommand{\aut}[0]{\mathcal{A}}
\newcommand{\autof}[1]{\aut_{#1}}
\newcommand{\lang}[0]{\mathcal{L}}
\newcommand{\langof}[1]{\lang(#1)}
\newcommand{\ltr}[1]{\xrightarrow{#1}}
\newcommand{\post}[0]{\mathit{Post}}
\newcommand{\postof}[2]{\post(#1, #2)}
\newcommand{\acc}[0]{\mathit{Acc}}

\newcommand{\fin}[0]{\mathit{Fin}}
\newcommand{\finof}[2]{\fin(#1, #2)}
\newcommand{\reach}[0]{\mathit{Reach}}
\newcommand{\residue}[0]{\kappa}

\newcommand{\vecof}[1]{\vec{#1}}
\newcommand{\modof}[1]{\mathrel{\equiv_{#1}}}
\newcommand{\vars}[0]{\mathbb{X}}
\newcommand{\asgn}[0]{\nu}
\newcommand{\modclof}[2]{\left[#1\right]_{#2}}
\newcommand{\gcdof}[1]{\gcd(#1)}

\newcommand{\formulaof}[1]{\tikz[anchor=base,baseline]{
  \node[inner sep=0.5mm, fill=blue!5, minimum height=4mm] {$#1$};}}
\newcommand{\true}[0]{\top}
\newcommand{\false}[0]{\bot}

\newcommand{\metatrue}[0]{\mathit{true}}
\newcommand{\metafalse}[0]{\mathit{false}}

\newcommand{\defeq}[0]{\mathrel{\stackrel{\text{def}}{=}}}
\newcommand{\defiff}[0]{\mathrel{\stackrel{\text{def}}{\Leftrightarrow}}}

\newcommand{\fv}[0]{\mathit{fv}}
\newcommand{\proj}[0]{\pi}
\newcommand{\alphx}[0]{\Sigma_{\vars}}

\newcommand{\structsim}[0]{\preceq_{\mathsf{s}}}

\newcommand{\dec}[0]{\mathsf{blw}}
\newcommand{\inc}[0]{\mathsf{abv}}
\newcommand{\decat}[0]{\dec_{\mathit{atom}}}
\newcommand{\incat}[0]{\inc_{\mathit{atom}}}
\newcommand{\atom}[0]{\mathsf{atom}}
\newcommand{\flow}[0]{\mathsf{FA}}

\newcommand{\meet}[0]{\sqcap}
\newcommand{\range}[0]{\mathsf{range}}
\newcommand{\rangeat}[0]{\range_{\mathit{atom}}}
\newcommand{\undefined}[0]{\mathit{undef}}
\newcommand{\datadom}[0]{\mathbb{D}}

\newcommand{\subst}[2]{[#1/#2]}

\newcommand{\lash}[0]{\textsc{Lash}\xspace}
\newcommand{\mona}[0]{\textsc{Mona}\xspace}
\newcommand{\amaya}[0]{\textsc{Tool}\xspace}
\newcommand{\amayanoopt}[0]{$\amaya_{\textrm{noopt}}$\xspace}
\newcommand{\sylvan}[0]{\textsc{Sylvan}\xspace}
\newcommand{\ziii}[0]{Z3\xspace}
\newcommand{\cvc}[0]{\textsc{cvc5}\xspace}
\newcommand{\princess}[0]{\textsc{Princess}\xspace}

\newcommand{\Rplus}{\protect\hspace{-.1em}\protect\raisebox{.35ex}{\smaller{\smaller\textbf{+}}}}
\newcommand{\Cpp}{\mbox{C\Rplus\Rplus}\xspace}
\newcommand{\C}{\mbox{C}\xspace}

\newcommand{\RGBcircle}[1]{\ensuremath{{\color[RGB]{#1}\bullet}}}

\newcommand{\worklist}[0]{\mathsf{worklist}}

\newcommand{\pop}[0]{\mathsf{pop}}

\newcommand{\clcoNP}[0]{\textbf{co-NP}}
\newcommand{\clEXPSPACE}[0]{\textbf{EXPSPACE}}
\newcommand{\clNEXP}[0]{\textbf{NEXP}}
\newcommand{\clEXP}[0]{\textbf{EXP}}

\tikzset{
  >={Stealth[round,bend]},
}
\tikzstyle{automaton}=[
  semithick,shorten >=0pt,>={Stealth[round,bend]},
  node distance=1.5cm,
  initial text=,
  every initial by arrow/.style={every node/.style={inner sep=0pt}},
  every state/.style={
    align=center,
    fill=white,
    minimum size=7.5mm,
    inner sep=0pt,
    execute at begin node=\strut,
  },%
]

\newcommand{\twosymb}[2]{\big[\begin{smallmatrix} #1 \\ #2 \end{smallmatrix}\big]}
\newcommand{\twosymbfin}[2]{\protect\tikz[anchor=base,baseline,inner sep=0mm]{\protect\node[fill=red!20]{$\twosymb{#1}{#2}$};}}
\newcommand{\onesymb}[1]{\begin{smallmatrix}\left[#1\right]\end{smallmatrix}}

\tikzstyle{labsymb}=[fill=white,fill opacity=0.70,inner sep=0mm]

\newcommand{\smtcomp}[0]{\texttt{SMT-COMP}\xspace}
\newcommand{\frobenius}[0]{\texttt{Frobenius}\xspace}

\InputIfFileExists{cutmargins.tex}{%
	\pdfpagesattr{/CropBox [125 82 490 688]} 
}{%
}

\Crefname{algocf}{Algorithm}{Algorithms}

\title{Algebraic Reasoning Meets Automata in Solving\\ Linear Integer Arithmetic (Technical Report)
}
\author{
Peter~Habermehl\orcidID{0000-0002-7982-0946}\inst{2} \and
Vojtěch~Havlena\orcidID{0000-0003-4375-7954}\inst{1} \and
Michal~Hečko\orcidID{0009-0003-2428-8547}\inst{1}\and
Lukáš~Holík\orcidID{0000-0001-6957-1651}\inst{1} \and
Ondřej~Lengál\orcidID{0000-0002-3038-5875}\inst{1}
}

\institute{
  Faculty of Information Technology,
Brno University of Technology, Brno, Czech Republic \and 
Université Paris Cité, IRIF, Paris, France}

\begin{document}

\maketitle

\begin{abstract}

We present a new angle on solving quantified linear integer arithmetic
based on combining the automata-based approach, where numbers are understood as
bitvectors, with ideas from (nowadays prevalent) algebraic approaches, which
work directly with numbers.
%
%
This combination is enabled by a~fine-grained version of the duality between automata and arithmetic formulae.
In particular, we employ a~construction where states of automaton are obtained as derivatives of arithmetic formulae: then every state corresponds to a~formula.
Optimizations based on techniques and ideas transferred from the world of algebraic methods are used on thousands of automata states, which dramatically amplifies their effect. 
The merit of this combination of automata with algebraic methods 
is demonstrated by our prototype implementation being competitive to and even
superior to state-of-the-art SMT solvers.

\end{abstract}

\vspace{-8.0mm}
\section{Introduction}\label{sec:introduction}
\vspace{-2.0mm}



\emph{Linear integer arithmetic} (LIA), also known as \emph{Presburger
arithmetic}, is the first-order theory of integers with addition.
Its applications include e.g.
databases~\cite{Kuper00}, program analysis~\cite{Monniaux09},
synthesis~\cite{KuncakMPS12}, 
and it is an essential component of every aspiring SMT solver. 
Many other types of constraints 
can either be
reduced to LIA, or are decided using a~tight collaboration of a~solver for the
theory and a~LIA solver, e.g., in~the~theory of bitvectors~\cite{SchuleS06},
strings~\cite{ChenCHHLS23}, or arrays~\cite{GhilardiNRZ07}.
Current SMT solvers are strong enough in~solving large quantifier-free LIA
formulae. 
Their ability to handle quantifiers is, however, problematic to the extent of being impractical.
Even a~tiny formula with two quantifier alternations can be a~show
stopper for them. 
Handling quantifiers is an~area of lively research with numerous application
possibilities waiting for a practical solution, e.g., software model
checking~\cite{HeizmannHP13}, program synthesis~\cite{ReynoldsKTBD19}, or
theorem proving~\cite{HieronymiMOS0S22}.

Among existing techniques for handling quantifiers, 
the complete approaches based on 
quantifier elimination~\cite{Presburger29,Cooper72} and
automata~\cite{Buchi60,WolperB95,BoudetC96} 
have been mostly deemed not scalable and abandoned in practice. 
Current SMT solvers use mainly incomplete techniques 
originating, e.g., from solving the theory of uninterpreted
functions~\cite{ReynoldsKK17} and algebraic techniques, such as the
\emph{simplex} algorithm for quantifier-free formulae~\cite{Dantzig60}.

This work is the first step in leveraging a recent renaissance of practically
competitive automata technology for solving LIA.
This trend that has recently emerged in string constraint solving (e.g.~\cite{noodler,noodlerfm,ostrich,trau,z3strre}), 
processing regular expressions \cite{cox,mata,VeanesDerivatives21},   
reasoning about the SMT theory of bitvectors~\cite{strejdaQBF}, 
or regex matching
(e.g.~\cite{oopslamatching,re2,hyperscan,margusmatcher}).
The new advances are rooted in paradigms such as
usage of non-determinism and alternation, 
various flavours of symbolic representations, 
and combination with/or integration into SAT/SMT frameworks and with algebraic techniques.  
%

We particularly show that the automata-based procedure provides unique
opportunities to amplify certain algebraic optimizations that reason over the
semantic of formulae.
These optimizations then boost the inherent strong points of the automata-based
approach to the extent that it is able to overcome modern SMT solvers.
%
%
The core strong points of automata are orthogonal to those of algebraic methods, 
mainly due to treating numbers as strings of bits regardless of their numerical values.
%
%
Automata can thus represent large sets of solutions succinctly and can use powerful techniques, such as minimization, that have no counterpart in the algebraic world. 
This makes automata more efficient than the algebraic approaches already in their basic form, implemented e.g. in~\cite{WolperB95,BoudetC96}, on some classes of problems such as the \emph{Frobenius coin problem}~\cite{Haase18}.
\begin{wrapfigure}[15]{r}{0.51\textwidth}
  \vspace{-6mm}
  \hspace*{-3mm}
  \begin{minipage}{65mm}
    \includegraphics[width=65mm]{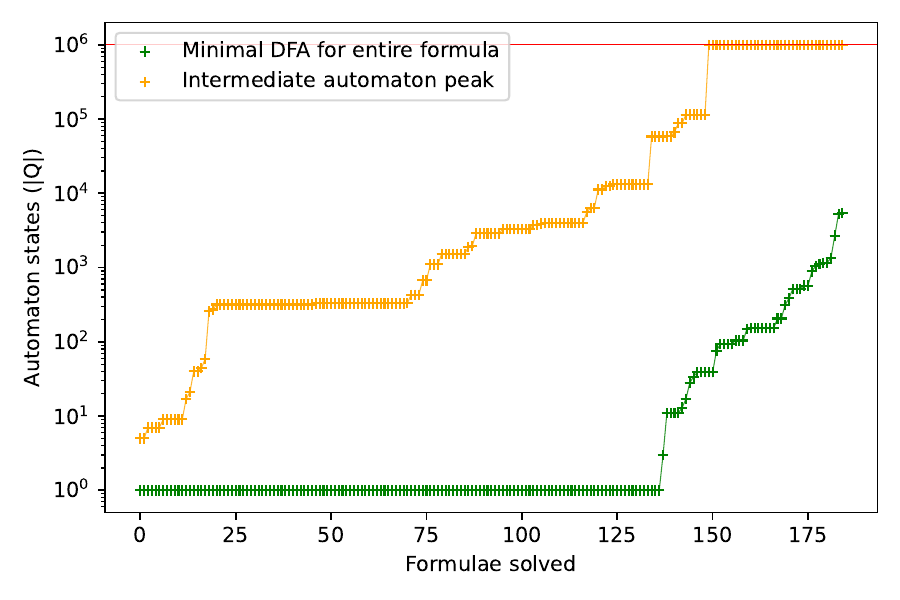}
  \end{minipage}
  \vspace{-4mm}
  \caption{
    Comparison of the peak intermediate automaton size and the size of the
    minimized DFA for the entire formula
    on the \smtcomp benchmark (cf.\ \cref{sec:experiments}).
  }
\label{fig:gap}
\end{wrapfigure}
In many practical cases, the automata construction, however, explodes.
The explosion usually happens when constructing an intermediate automaton for
a~sub-formula, although the minimal automaton for the entire formula is almost
always small.
The plot in \cref{fig:gap} shows that the gap between sizes of final and
intermediate automata in our benchmark is always several orders of magnitude large,
offering opportunities for optimizations.
%
%
In this paper, we present a~basic approach to breaching this gap by transferring techniques and ideas from the algebraic world to automata and using them to prune the vast state space. 

To this end, we combine the classical inductive automata construction with constructing formula derivatives, 
similar to derivatives of regular expressions \cite{Brzozowski64_derivatives,Antimirov96_partderivatives,VeanesDerivatives21} or WS1S/WS$k$S formulae~\cite{traytel,FiedorHJLV17,HavlenaHLV19}. 
Our construction directly generates states of an automaton of a nested formula, 
without the need to construct intermediate automata for sub-formulae first.
Although the derivative construction is not better than the inductive construction by itself, 
it gives an opportunity to optimize the state space \emph{on the fly},
before it gets a chance to explode. 
The optimization itself is negotiated by the \emph{fine-grained} version of the well-known \emph{automaton-formula duality}. 
In the derivative construction, \emph{every state} corresponds to a~LIA formula.
Applying equivalence-preserving formula rewriting on state formulae has
the effect of merging or pruning states, similar to what
DFA minimization could achieve after the entire automaton were constructed.

Our equivalence-preserving rewriting uses known algebraic techniques or ideas
originating from them.
%
%
%
First, we use basic formula simplification techniques, such as propagating true or false values or antiprenexing. 
Despite being simple, these simplifications have a large impact on performance.
Second, we use \emph{disjunction pruning}, which
replaces $\varphi_1 \lor \varphi_2 \lor \cdots \lor \varphi_k$
by $\varphi_2 \lor \cdots \lor \varphi_k$ if $\varphi_1$ is
entailed by the rest of the formula (this is close to the state pruning
techniques used in \cite{fivedeterm,FiedorHJLV17,DoyenR10}).
%
We also adopt the principle of \emph{quantifier
instantiation}~\cite{simplify,reynoldsConflictQI,geModelBasedQI},
where we detect cases when a~quantified variable can be substituted by one or several values, or when
a~linear congruence can be simplified to a~linear equation. We particularly use ideas from Cooper's quantifier elimination~\cite{Cooper72},
where a quantifier is expanded into a disjunction over a finite number of values, 
and from Omega test~\cite{Pugh91}, where a~variable with a one-side unbounded
range is substituted by the least restrictive value.

It is noteworthy that in the purely algebraic setting,
the same techniques could only be applied once on the input formula, with a negligible effect.
In the automata-based procedure, their power is amplified since
they are used on thousands of derivative states generated deep within automata after reading several bits of the solution.

Our prototype implementation is competitive with the best SMT solvers on benchmarks from SMT-COMP,
and, importantly, it is superior on quantifier-intensive instances.
We believe that more connections along the outlined direction, based on the
fine-grained duality between automata and formulae, can be found, and that the
work in this paper is the first step in bridging the worlds of automata
and algebraic approaches.
Many challenges in incorporating automata-based LIA reasoning into SMT solvers
still await but, we believe, can be tackled, as witnessed e.g.\ within the
recent successes of the integration of automata-based string
solvers~\cite{ChenCHHLS23,z3strre,ostrich}.


%

%

\vspace{-3.0mm}
\section{Preliminaries}\label{sec:prelims}
\vspace{-2.0mm}

We use~$\integers$ to denote the set of \emph{integers}, $\integers^+$ to
denote the set of \emph{positive integers}, and~$\bin$ to denote the set of
\emph{binary digits} $\{0,1\}$.
For $x, y \in \integers$ and $m \in \integers^+$, we use $x \modof m y$ to
denote that~$x$ is congruent with~$y$ modulo~$m$, i.e., there exists~$z \in
\integers$ s.t.\ $z \cdot m + x = y$; and~$x | y$ to denote that there
exists $z' \in \integers$ s.t.\ $y = z' \cdot x$.
Furthermore, we use~$\modclof x m$ to denote the unique integer
s.t.\ $0 \leq \modclof x m < m$ and $x \modof m \modclof x m$.
The following notation will be used for intervals of integers:
for $a,b \in \integers$, the set $\{x \in \integers \mid a \leq x \leq b\}$ is denoted as $[a,b]$,
the set $\{x \in \integers \mid a \leq x\}$ is denoted as $[a, {+\infty})$, and
the set $\{x \in \integers \mid x \leq b\}$ is denoted as $({-\infty}, b]$.
The \emph{greatest common divisor} of $a,b \in \integers$, denoted as
$\gcdof{a,b}$, is the largest integer such that $\gcdof{a,b}|a$ and
$\gcdof{a,b}|b$ (note that $\gcdof{a,0} = |a|$); if~$\gcdof{a,b} = 1$, we
say that~$a$ and~$b$ are \emph{coprime}.
For a~real number~$y$, $\lfloor y \rfloor$~denotes the \emph{floor} of~$y$,
i.e., the integer $\max\{z \in \integers \mid z \leq y\}$, and
$\lceil y \rceil$~denotes the \emph{ceiling} of~$y$,
i.e., the integer $\min\{z \in \integers \mid z \geq y\}$.

An~\emph{alphabet}~$\Sigma$ is a~finite non-empty set of \emph{symbols} and
a~\emph{word} $w = a_1 \ldots a_n$ of length~$n$ over~$\Sigma$ is a~finite
sequence of symbols from~$\Sigma$. 
If~$n = 0$, we call~$w$ the \emph{empty word} and denote it~$\epsilon$.
$\Sigma^+$~is the \mbox{set of all non-empty words over~$\Sigma$ and $\Sigma^* =
\Sigma^+ \cup \{\epsilon\}$.}

\vspace{-2mm}
\paragraph{Finite automata.}
In order to simplify constructions in the paper, we use a variation of finite
automata with accepting transitions instead of states.
A~(\emph{final transition acceptance-based) nondeterministic finite automaton} (FA) is a~five-tuple $\aut = (Q,
\Sigma, \delta, I, \acc)$ where
$Q$ is a~finite set of \emph{states},
$\Sigma$ is an~\emph{alphabet},
$\delta \subseteq Q \times \Sigma \times Q$ is a~\emph{transition relation},
$I \subseteq Q$ is a~set of \emph{initial states}, and
$\acc\colon \delta \to \{\metatrue, \metafalse\}$ is a transition-based acceptance condition. 
We often use $q \ltr a p$ to denote that $(q,a,p) \in \delta$.
A~\emph{run} of~$\aut$ over a~word $w = a_1 \ldots a_n$ is a~sequence of states
$\rho = q_0 q_1 \ldots q_n \in Q^{n+1}$ such that for all $1 \leq i \leq n$ it
holds that $q_{i-1} \ltr{a_i} q_i$ and $q_0 \in I$.
The run~$\rho$ is \emph{accepting} if $n \geq 1$ and $\acc(q_{n-1} \ltr{a_n} q_n)$
(i.e., if the last transition in the run is accepting)%
\footnote{
  Note that our FAs cannot accept the empty word~$\epsilon$, which corresponds in our use to
  the fact that in the two's complement encoding of integers, one needs at
  least one bit (the sign bit) to represent a number, see further.
}.
The language of~$\aut$, denoted as~$\langof \aut$, is defined as $\langof \aut
= \{w \in \Sigma^* \mid \text{there is an accepting run of } \aut \text{ on }
w\}$.
We further use $\lang_\aut(q)$ to denote the language of the FA obtained
from~$\aut$ by setting its set of initial states to~$\{q\}$ (if the context is
clear, we use just~$\langof q$).

$\aut$ is \emph{deterministic} (a~DFA) if $|I| \leq 1$ and for all
states~$q \in Q$ and symbols~$a \in \Sigma$, it holds that if $q \ltr a p$ and
$q \ltr a r$, then $p=r$.
On the other hand, $\aut$~is \emph{complete} if $|I| \geq 1$ and for all
states~$q \in Q$ and symbols $a \in \Sigma$, there is at least one state~$p \in
Q$ such that $q \ltr a p$.
For a~deterministic and complete~$\aut$, we abuse notation and treat~$\delta$
as a~function $\delta\colon Q \times \Sigma \to Q$.
A~DFA~$\aut$ is \emph{minimal} if $\forall q \in Q\colon \langof
q \neq \emptyset \land \forall p \in Q\colon p \neq q \Rightarrow \langof q
\neq \langof p$.
Hopcroft's~\cite{Hopcroft71} and Brzozowski's~\cite{Brzozowski63} algorithms
for obtaining a~minimal DFA can be modified for our definition of FAs .

\vspace{-2mm}
\paragraph{Linear integer arithmetic.}
Let~$\vars = \{x_1, \ldots, x_n\}$ be a~(finite) set of integer variables.
We will use $\vecof x$ to denote the vector $(x_1, \ldots, x_n)$.
Sometimes, we will treat $\vecof x$ as a~set, e.g., $y \in \vecof x$ denotes $y
\in \{x_1, \ldots, x_n\}$.
A~\emph{linear integer arithmetic} (LIA) formula~$\formulaof\varphi$ over~$\vars$ is
obtained using the following grammar:
\begin{align*}
  \formulaof{\varphi_{\mathit{atom}}} ::= \quad &
      \formulaof{\vecof a \cdot \vecof x = c} ~\mid~
      \formulaof{\vecof a \cdot \vecof x \leq c} ~\mid~
      \formulaof{\vecof a \cdot \vecof x \modof m c} ~\mid~
      \formulaof{\bot}
      \\
  \formulaof{\varphi} ::= \quad &
      \formulaof{\varphi_{\mathit{atom}}} ~\mid~
      \formulaof{\neg \varphi} ~\mid~
      \formulaof{\varphi \land \varphi} ~\mid~
      \formulaof{\varphi \lor \varphi} ~\mid~
      \formulaof{\exists y (\varphi)}
\end{align*}
where~$\vecof a$ is a~vector of~$n$ integer coefficients $(a_1, \ldots, a_n)
\in \integers^n$, $c \in \integers$~is a~constant, $m\in \integers^+$~is
a~\emph{modulus}, and $y \in \vars$
(one can derive the other connectives $\true$, $\rightarrow$,
$\leftrightarrow$, $\forall$, \ldots in the standard way)%
\footnote{
  Although the modulo constraint $\formulaof{\vecof a \cdot \vecof x \modof m c}$ could be
  safely removed without affecting the expressivity of the input language,
  keeping it allows a more efficient automata construction and application of
  certain heuristics (cf.\ \cref{sec:monotonicity}).
}.
Free variables of~$\formulaof \varphi$ are denoted as~$\fv(\formulaof \varphi)$.
Given a~formula~$\formulaof \varphi$, we say that an assignment $\asgn\colon \vars \to
\integers$ is a~\emph{model} of~$\formulaof \varphi$, denoted as $\asgn \models \formulaof \varphi$,
if~$\asgn$ satisfies~$\formulaof\varphi$ in the standard way.
Note that we use the same symbols $=, \leq, \modof m, \neg, \land, \lor,
\exists, \ldots$ in the syntactical language (where they are not to be
interpreted, with the exception of evaluation of constant expressions) of the logic as
well as in the meta-language.
In order to avoid ambiguity, we use the style $\formulaof \varphi$ for a~syntactic formula.
W.l.o.g.\ we assume that variables in~$\formulaof \varphi$ are unique, i.e.,
there is no overlap between quantified variables and also between free and
quantified variables.

In our decision procedure we represent integers as non-empty sequences of
binary digits $a_0 \ldots a_n \in \bin^+$ using the two's complement with the
\emph{least-significant bit first} (LSBF) encoding (i.e., the right-most bit
denotes the \emph{sign}).
Formally, the \emph{decoding} of a~binary word represents the integer
\begin{equation}
\label{eq:dec}
\modencof{a_0 \ldots a_n} = \sum_{i = 0}^{n-1} a_i \cdot 2^i - a_n \cdot 2^n.
\end{equation}
For instance, $\decof{0101} = -6$ and $\decof{010} = 2$.
Note that any integer has infinitely many representations in this encoding: the
shortest one and others obtained by repeating the sign bit any number of times.
In this paper, we work with the so-called \emph{binary assignments}.
A~binary assignment is an assignment $\asgn\colon \vars \to \bin^+$ s.t. for each $x_1, x_2 \in \vars$ the lengths of the words assigned to~$x_1$ and~$x_2$ match, i.e., $|\asgn(x_1)| = |\asgn(x_2)|$.
We overload the decoding operator $\modencof \cdot$ to binary assignments such that
$\modencof \asgn \colon \vars
\to \integers$ is defined as $\modencof \asgn = \{ x \mapsto \modencof y \mid \asgn(x) = y \}$.
%
A~\emph{binary model} of a~formula~$\formulaof \varphi$ is a~binary assignment~$\asgn$
such that $\modencof \asgn \models \varphi$.
We~denote the set of all binary models of 
a~LIA formula~$\formulaof \varphi$ as~$\semof{\formulaof \varphi}$ and
we write $\formulaof{\varphi_1} \Rightarrow \formulaof{\varphi_2}$ to denote
$\semof{\formulaof{\varphi_1}} \subseteq \semof{\formulaof{\varphi_2}}$ and
$\formulaof{\varphi_1} \Leftrightarrow \formulaof{\varphi_2}$ to denote
$\semof{\formulaof{\varphi_1}} = \semof{\formulaof{\varphi_2}}$.


\newcommand{\figPostAtoms}[0]{ 
\begin{figure}[t]
\begin{center}
\newlength{\longestPostFormula}
\settowidth{\longestPostFormula}{$\formulaof{\vecof a \cdot \vecof x \modof {2m+1} \modclof{\frac 1 2 (\residue + 2m + 1)}{2m+1}}$}
\begin{align*}
  \postof{\formulaof{\vecof a \cdot \vecof x \leq c}}{\sigma} ~ &\textstyle \defeq~
  \formulaof{\vecof a \cdot \vecof x \leq \lfloor\frac 1 2 \residue \rfloor}\hspace{10mm} \text{for }~\residue \defeq c - \vecof a \cdot \sigma\\
  \postof{\formulaof{\vecof a \cdot \vecof x = c}}{\sigma} ~ &\textstyle \defeq~
    \begin{cases}
      \makebox[\longestPostFormula][l]{$\formulaof{\vecof a \cdot \vecof x = \frac 1 2 \residue}$} & \text{if } 2 | \residue\\
      \formulaof\bot & \text{otherwise }
    \end{cases} \\
  \postof{\formulaof{\vecof a \cdot \vecof x \modof {2m} c}}{\sigma} ~ &{} \defeq~
    \begin{cases}
    \makebox[\longestPostFormula][l]{$\formulaof{\vecof a \cdot \vecof x \modof {m} \modclof{\frac 1 2 \residue} m}$} & \text{if } 2 | \residue \\
      \formulaof{\bot} & \text{otherwise }
    \end{cases} \\
  \postof{\formulaof{\vecof a \cdot \vecof x \modof {2m + 1} c}}{\sigma} ~ &{} \defeq~
    \begin{cases}
      \formulaof{\vecof a \cdot \vecof x \modof {2m+1} \modclof{\frac 1 2 \residue}{2m+1}} & \text{if } 2 | \residue \\
      \formulaof{\vecof a \cdot \vecof x \modof {2m+1} \modclof{\frac 1 2 (\residue + 2m + 1)}{2m+1}} & \text{otherwise}
    \end{cases} \\
  \postof{\formulaof{\bot}}{\sigma} ~ &{} \defeq~ \formulaof{\bot}
\end{align*}
%
\end{center}
\vspace{-5mm}
\caption{
  Definition of the transition function $\post$ for atomic formulae.
  Note that the right-hand sides contain constant expressions, so they will be
  evaluated.
 }
\label{fig:post_atoms}
\vspace{-3mm}
\end{figure}
}

\newcommand{\figFinAtoms}[0]{ 
\begin{wrapfigure}[8]{r}{5.6cm}
\vspace*{-10mm}
\hspace*{-5mm}
\begin{minipage}{1.15\linewidth}
\begin{align*}
  \finof{\formulaof{\vecof a \cdot \vecof x \leq c}}{\sigma} ~ &\textstyle \defiff~
    c + \vecof a \cdot \sigma \geq 0 \\[-1mm]
  \finof{\formulaof{\vecof a \cdot \vecof x = c}}{\sigma} ~ &\textstyle \defiff~
    c + \vecof a \cdot \sigma = 0\\[-1mm]
  \finof{\formulaof{\vecof a \cdot \vecof x \modof m c}}{\sigma} ~ &{} \defiff~
    c + \vecof a \cdot \sigma \modof m 0\\[-1mm]
  \finof{\formulaof{\bot}}{\sigma} ~ &{} \defiff~ \metafalse
\end{align*}
\vspace{-9mm}
\caption{Acceptance for atomic formulae.}
\label{fig:fin_atoms}
\end{minipage}
\end{wrapfigure}
}

\newcommand{\algAutConstr}[0]{ 
\begin{figure}[t]
\begin{algorithm}[H]
\caption{Atomic formulae automata construction.}
\KwIn{atomic formula~$\formulaof \varphi$}
\KwOut{FA $\autof \varphi$ such that $\langof{\autof \varphi} = \semof{\formulaof \varphi}$}
\label{alg:aut_constr}
$Q \gets \worklist \gets \{\formulaof \varphi\}$\;
\While{$\worklist \neq \emptyset$}{
  $\formulaof \psi \gets \worklist.\pop()$\;
  \ForEach{$a \in \alphx$}{
    $\formulaof{\psi'} \gets \postof{\formulaof \psi}{a}$\;
    \If{$\formulaof{\psi'} \notin Q$}{
      $Q \gets Q \cup \{\formulaof{\psi'}\}$\;
      $\worklist \gets \worklist \cup \{\formulaof{\psi'}\}$\;
    }
  }
}
\Return{$\autof \varphi = (Q, \alphx, \post \ol{restr?}, \{\formulaof \varphi\}, )$}\;
\end{algorithm}
\end{figure}
}

\newcommand{\figExampleAtoms}[0]{   
\begin{figure}[t]
\hfill
\begin{subfigure}[b]{0.45\linewidth}
\begin{center}
\begin{tikzpicture}
[
  automaton,
  node distance=20mm
]

\node[initial,state] (q1) {$1$};
\node[state,below of=q1] (q0) {$0$};
\node[state,right of=q1] (q-1) {$-1$};
\node[state,right of=q0] (q-2) {$-2$};
\node[state,right of=q-2] (q-3) {$-3$};

\draw[->] (q1) to node[centered,labsymb,pos=0.3] {$\twosymbfin 0 0$}
                  node[centered,labsymb,pos=0.7] {$\twosymbfin 1 0$} (q0);
\draw[->] (q1) to node[centered,labsymb,pos=0.3] {$\twosymbfin 0 1$}
                  node[centered,labsymb,pos=0.7] {$\twosymbfin 1 1$} (q-1);

\draw[->,loop below] (q0) to node[centered,labsymb] {$\twosymbfin 0 0$} (q0);
\draw[->] (q0) to node[centered,labsymb,pos=0.35] {$\twosymbfin 0 1$} 
                  node[centered,labsymb,pos=0.65] {$\twosymbfin 1 0$} (q-1);
\draw[->] (q0) to node[centered,labsymb] {$\twosymbfin 1 1$} (q-2);

\draw[->,loop right] (q-1) to node[centered,labsymb,pos=0.3,yshift=1mm] {$\twosymb 0 0$}
                              node[centered,labsymb,pos=0.7,yshift=-1mm] {$\twosymbfin 1 0$} (q-1);
\draw[->,bend left] (q-1) to node[centered,labsymb,pos=0.3] {$\twosymbfin 0 1$}
                             node[centered,labsymb,pos=0.6] {$\twosymbfin 1 1$} (q-2);

\draw[->,loop below] (q-2) to node[centered,labsymb,pos=0.3,xshift=1mm] {$\twosymbfin 0 1$}
                              node[centered,labsymb,pos=0.7,xshift=-1mm] {$\twosymb 1 0$} (q-2);
\draw[->,bend left] (q-2) to node[centered,labsymb] {$\twosymb 0 0$} (q-1);
\draw[->,bend left] (q-2) to node[centered,labsymb] {$\twosymbfin 1 1$} (q-3);

\draw[->,loop below] (q-3) to node[centered,labsymb,pos=0.3,xshift=1mm] {$\twosymb 0 1$}
                              node[centered,labsymb,pos=0.7,xshift=-1mm] {$\twosymbfin 1 1$} (q-3);
\draw[->,bend left] (q-3) to node[centered,labsymb,pos=0.3] {$\twosymb 0 0$}
                             node[centered,labsymb,pos=0.7] {$\twosymb 1 0$} (q-2);

\end{tikzpicture}
\end{center}
\vspace{-4mm}
\caption{The FA for $\formulaof{x + 2y \leq 1}$. A~state $\formulaof{x + 2y \leq
  c}$ is represented by ``$c$''.}
\label{label}
\end{subfigure}
\hfill
\begin{subfigure}[b]{0.40\linewidth}
\begin{center}
\begin{tikzpicture}
[
  automaton,
  node distance=20mm
]

\node[initial,state] (q2_6) {$2_{\modof 6}$};
\node[state,right of=q2_6] (q1_3) {$1_{\modof 3}$};
\node[state,below of=q2_6] (q0_3) {$0_{\modof 3}$};
\node[state,below of=q1_3] (q2_3) {$2_{\modof 3}$};

\draw[->] (q2_6) to node[centered,labsymb] {$\twosymb 0 0$} (q1_3);
\draw[->] (q2_6) to node[centered,labsymb] {$\twosymb 0 1$} (q0_3);

\draw[->,loop right] (q1_3) to node[centered,labsymb] {$\twosymbfin 0 1$} (q1_3);
\draw[->,bend left=15] (q1_3) to node[centered,labsymb] {$\twosymb 1 0$} (q0_3);
\draw[->,bend left] (q1_3) to node[centered,labsymb,pos=0.3] {$\twosymb 0 0$}
                              node[centered,labsymb,pos=0.7] {$\twosymb 1 1$} (q2_3);

\draw[->,loop left] (q0_3) to node[centered,labsymb,pos=0.3,yshift=-1mm] {$\twosymbfin 0 0$}
                              node[centered,labsymb,pos=0.7,yshift=1mm] {$\twosymbfin 1 1$} (q0_3);
\draw[->,bend left] (q0_3) to node[centered,labsymb] {$\twosymb 1 0$} (q1_3);
\draw[->,bend left=15] (q0_3) to node[centered,labsymb] {$\twosymb 0 1$} (q2_3);

\draw[->,loop right] (q2_3) to node[centered,labsymb] {$\twosymbfin 1 0$} (q2_3);
\draw[->,bend left=15] (q2_3) to node[centered,labsymb,pos=0.3] {$\twosymb 0 0$}
                              node[centered,labsymb,pos=0.7] {$\twosymb 1 1$} (q1_3);
\draw[->,bend left] (q2_3) to node[centered,labsymb] {$\twosymb 0 1$} (q0_3);

\end{tikzpicture}
\end{center}
\vspace{-4mm}
\caption{The FA for $\formulaof{x + 2y \modof 6 2}$. A~state $\formulaof{x + 2y \modof m
  c}$ is represented by ``$c_{\modof m}$''.}
\label{label}
\end{subfigure}
\hfill
\vspace{-2mm}
\caption{Examples of FAs for atomic formulae.
  The notation for symbols is $\twosymb x y$; red background denotes accepting
  transitions.
  }
\label{fig:example_atoms}
\vspace{-4mm}
\end{figure}
}

\vspace{-0.0mm}
\section{Classical Automata-Based Decision Procedure for LIA}\label{sec:classical}
\vspace{-0.0mm}

The following \emph{classical decision procedure} is due to Boudet and
Comon~\cite{BoudetC96} (based on the ideas of \cite{Buchi62}) with
an extension to modulo constraints by
Durand-Gasselin and Habermehl~\cite{Durand-GasselinH10}.
Given a~set of variables~$\vars$, a~symbol $\sigma$ is a~mapping \mbox{$\sigma\colon
\vars \to \bin$} and $\alphx$~denotes the set of all symbols over~$\vars$.
For a~symbol $\sigma \in \alphx$ and a~variable
$x\in\vars$ we define the \emph{projection} $\proj_x(\sigma) = \{ \sigma'\in
\Sigma_{\vars} \mid \sigma'_{|\vars\setminus \{x\}} =
\sigma_{\left|\vars\setminus \{x\}\right.} \}$ where $\sigma_{|\vars \setminus \{x\}}$ is the
restriction of the function~$\sigma$ to the domain~$\vars \setminus \{x\}$.

For a~LIA formula~$\varphi$, the classical automata-based decision procedure builds
an FA $\aut_\varphi$ encoding all binary models of~$\varphi$.
We use a modification which uses automata with accepting edges instead of states.
It allows to construct deterministic automata for atomic formulae, later in \cref{sec:opt} also for complex formulae, 
and to eliminate an artificial final state present in the original construction that does not correspond to any arithmetic formula. 
The construction proceeds inductively as follows:

\figPostAtoms 

\paragraph{Base case.}
First, an~FA $\autof{\varphi_{\mathit{atom}}}$ is constructed for each
atomic formula~$\varphi_{\mathit{atom}}$ in~$\varphi$.
The states of~$\autof{\varphi_{\mathit{atom}}}$ are LIA formulae with
$\formulaof{\varphi_{\mathit{atom}}}$ being the (only) initial state.
$\autof{\varphi_{\mathit{atom}}}$'s structure is given by the
transition function $\post$, implemented via a~derivative
$\post(\formulaof{\varphi_{\mathit{atom}}}, \sigma)$ of
$\varphi_{\mathit{atom}}$ w.r.t.\ symbols~$\sigma \in
\alphx$ as given in \cref{fig:post_atoms} (an example will follow).

Intuitively, for $\post(\formulaof{\vecof a \cdot \vecof x = c}, \sigma)$, the
next state after reading~$\sigma$ is given by taking the least significant bits (LSBs)
of all variables ($\vecof x$) after being multiplied with the respective
coefficients ($\vecof a$) and subtracting this value from~$c$.
If the parity of the result is odd, we can terminate ($\vecof a \cdot \vecof x$
and~$c$ have a~different LSB, so they cannot match),
otherwise we can remove the LSB of the result, set it as
a~new~$c$, and continue.
One can imagine this process as performing a~long addition of several binary
numbers at once with~$c$ being the result (the subtraction from~$c$ can be seen
as working with \emph{carry}).
The intuition for a~formula $\formulaof{\vecof a \cdot \vecof x \leq c}$ is
similar.
On the other hand,
for a~formula $\formulaof{\vecof a \cdot \vecof x \modof{2m} c}$, i.e.,
a~congruence with an even modulus, if the parity of the
left-hand side ($\vecof a \cdot \vecof x$) and the right-hand side~($c$) does
not match (in other words, $c - \vecof a \cdot \vecof x$ is odd),
because the modulus is even we can terminate.
Otherwise, we remove the LSB of the modulus (i.e., divide it by two).
Lastly, let us mention the second case for the rule for a~formula of the form
$\formulaof{\vecof a \cdot \vecof x \modof{2m+1} c}$.
Here, since~$\residue$ is odd, we cannot divide it by two; however, adding the
modulus~($2m+1$) to~$\residue$ yields an even value equivalent to~$\residue$.

The states of $\aut_{\varphi_{\mathit{atom}}}$ are then all reachable formulae obtained from the
application of $\post$ from the initial state, with the reachability from
a~set of formulae~$S$ using symbols from~$\Gamma$ given using the least
fixpoint operator~$\mu$ as follows:
\begin{equation}
  \reach(S,\Gamma) = \mu Z\colon S \cup \{\post(\formulaof{\psi},a) \mid \formulaof{\psi} \in Z, a \in \Gamma\}
\end{equation}

\begin{lemma}
  $\reach(\{\formulaof{\varphi_{\mathit{atom}}}\}, \alphx)$ is finite
  for an atomic formula $\formulaof{\varphi_{\mathit{atom}}}$.
\end{lemma}

\begin{proof}
  The cases for linear equations and inequations follow
  from~\cite[Proposition~1]{BoudetC96} and \cite[Proposition~3]{BoudetC96}
  respectively.
  For moduli, the lemma follows from the fact that in the definition
  of~$\post$, the right-hand side of a~modulo is an integer
  from~$[0,m-1]$.
\qed
\end{proof}



\figFinAtoms  

$\post$ is deterministic, so it suffices to define the acceptance condition
for the derivatives only for each state and symbol, as given in
\cref{fig:fin_atoms}.
E.g., a~transition from $\formulaof{2 x_1 - 7 x_2 = 5}$ over $\sigma = \twosymb
1 1$ is accepting; the intuition is similar as for $\post$
with the difference that the last bit is the sign bit (cf.~\cref{eq:dec}), so it
is treated in the opposite way to other bits (therefore, there is the ``$+$'' sign on the right-hand sides of the definitions rather than the ``$-$'' sign as in \cref{fig:post_atoms}).
If we substitute into the example, we obtain $2 \cdot (-1) - 7 \cdot (-1) = -2 +7 = 5$.
The acceptance condition $\acc$ is then defined as $\acc(\formulaof{\varphi_1}
\ltr{\sigma} \postof{\formulaof{\varphi_1}}{\sigma}) \defeq \finof{\formulaof{\varphi_1}}{\sigma}$ and
$\aut_{\varphi_{\mathit{atom}}}$ is defined as the~FA
\begin{equation}
\autof{\varphi_{\mathit{atom}}} = (\reach(\{\formulaof{\varphi_{\mathit{atom}}}\}, \alphx),
  \alphx, \post, \{\formulaof{\varphi_{\mathit{atom}}}\}, \acc).
\end{equation}
%
Note that if an FA accepts a~word~$w$, it also accepts all words obtained by
appending any number of copies of the most significant bit (the sign) to~$w$.

\begin{example}
\cref{fig:example_atoms} gives examples of FAs for
$\formulaof{x + 2y \leq 1}$ and $\formulaof{x + 2y \modof 6 2}$.
\qed
\end{example}

\begin{figure}[t]
\hfill
\begin{subfigure}[b]{0.45\linewidth}
\begin{center}
\begin{tikzpicture}
[
  automaton,
  node distance=20mm
]

\node[initial,state] (q1) {$1$};
\node[state,below of=q1] (q0) {$0$};
\node[state,right of=q1] (q-1) {$-1$};
\node[state,right of=q0] (q-2) {$-2$};
\node[state,right of=q-2] (q-3) {$-3$};

\draw[->] (q1) to node[centered,labsymb,pos=0.3] {$\twosymbfin 0 0$}
                  node[centered,labsymb,pos=0.7] {$\twosymbfin 1 0$} (q0);
\draw[->] (q1) to node[centered,labsymb,pos=0.3] {$\twosymbfin 0 1$}
                  node[centered,labsymb,pos=0.7] {$\twosymbfin 1 1$} (q-1);

\draw[->,loop below] (q0) to node[centered,labsymb] {$\twosymbfin 0 0$} (q0);
\draw[->] (q0) to node[centered,labsymb,pos=0.35] {$\twosymbfin 0 1$} 
                  node[centered,labsymb,pos=0.65] {$\twosymbfin 1 0$} (q-1);
\draw[->] (q0) to node[centered,labsymb] {$\twosymbfin 1 1$} (q-2);

\draw[->,loop right] (q-1) to node[centered,labsymb,pos=0.3,yshift=1mm] {$\twosymb 0 0$}
                              node[centered,labsymb,pos=0.7,yshift=-1mm] {$\twosymbfin 1 0$} (q-1);
\draw[->,bend left] (q-1) to node[centered,labsymb,pos=0.3] {$\twosymbfin 0 1$}
                             node[centered,labsymb,pos=0.6] {$\twosymbfin 1 1$} (q-2);

\draw[->,loop below] (q-2) to node[centered,labsymb,pos=0.3,xshift=1mm] {$\twosymbfin 0 1$}
                              node[centered,labsymb,pos=0.7,xshift=-1mm] {$\twosymb 1 0$} (q-2);
\draw[->,bend left] (q-2) to node[centered,labsymb] {$\twosymb 0 0$} (q-1);
\draw[->,bend left] (q-2) to node[centered,labsymb] {$\twosymbfin 1 1$} (q-3);

\draw[->,loop below] (q-3) to node[centered,labsymb,pos=0.3,xshift=1mm] {$\twosymb 0 1$}
                              node[centered,labsymb,pos=0.7,xshift=-1mm] {$\twosymbfin 1 1$} (q-3);
\draw[->,bend left] (q-3) to node[centered,labsymb,pos=0.3] {$\twosymb 0 0$}
                             node[centered,labsymb,pos=0.7] {$\twosymb 1 0$} (q-2);

\end{tikzpicture}
\end{center}
\vspace{-4mm}
\caption{The FA for $\formulaof{x + 2y \leq 1}$. A~state $\formulaof{x + 2y \leq
  c}$ is represented by ``$c$''.}
\label{label}
\end{subfigure}
\hfill
\begin{subfigure}[b]{0.40\linewidth}
\begin{center}
\begin{tikzpicture}
[
  automaton,
  node distance=20mm
]

\node[initial,state] (q2_6) {$2_{\modof 6}$};
\node[state,right of=q2_6] (q1_3) {$1_{\modof 3}$};
\node[state,below of=q2_6] (q0_3) {$0_{\modof 3}$};
\node[state,below of=q1_3] (q2_3) {$2_{\modof 3}$};

\draw[->] (q2_6) to node[centered,labsymb] {$\twosymb 0 0$} (q1_3);
\draw[->] (q2_6) to node[centered,labsymb] {$\twosymb 0 1$} (q0_3);

\draw[->,loop right] (q1_3) to node[centered,labsymb] {$\twosymbfin 0 1$} (q1_3);
\draw[->,bend left=15] (q1_3) to node[centered,labsymb] {$\twosymb 1 0$} (q0_3);
\draw[->,bend left] (q1_3) to node[centered,labsymb,pos=0.3] {$\twosymb 0 0$}
                              node[centered,labsymb,pos=0.7] {$\twosymb 1 1$} (q2_3);

\draw[->,loop left] (q0_3) to node[centered,labsymb,pos=0.3,yshift=-1mm] {$\twosymbfin 0 0$}
                              node[centered,labsymb,pos=0.7,yshift=1mm] {$\twosymbfin 1 1$} (q0_3);
\draw[->,bend left] (q0_3) to node[centered,labsymb] {$\twosymb 1 0$} (q1_3);
\draw[->,bend left=15] (q0_3) to node[centered,labsymb] {$\twosymb 0 1$} (q2_3);

\draw[->,loop right] (q2_3) to node[centered,labsymb] {$\twosymbfin 1 0$} (q2_3);
\draw[->,bend left=15] (q2_3) to node[centered,labsymb,pos=0.3] {$\twosymb 0 0$}
                              node[centered,labsymb,pos=0.7] {$\twosymb 1 1$} (q1_3);
\draw[->,bend left] (q2_3) to node[centered,labsymb] {$\twosymb 0 1$} (q0_3);

\end{tikzpicture}
\end{center}
\vspace{-4mm}
\caption{The FA for $\formulaof{x + 2y \modof 6 2}$. A~state $\formulaof{x + 2y \modof m
  c}$ is represented by ``$c_{\modof m}$''.}
\label{label}
\end{subfigure}
\hfill
\vspace{-2mm}
\caption{Examples of FAs for atomic formulae.
  The notation for symbols is $\twosymb x y$; red background denotes accepting
  transitions.
  }
\label{fig:example_atoms}
\vspace{-4mm}
\end{figure}

\paragraph{Inductive case.}
The inductive cases for Boolean connectives are defined in the standard way: conjunction of two formulae
is implemented by taking the intersection of the two corresponding FAs,
disjunction by taking their union, and negation is
implemented by taking the complement (which may involve determinization via the
subset construction).
Formally,
let $\autof{\varphi_i} = (Q_{\varphi_i},\alphx,
\delta_{\varphi_i}, I_{\varphi_i}, \acc_{\varphi_i})$ for $i \in \{1,2\}$ with
$Q_{\varphi_1} \cap Q_{\varphi_2} = \emptyset$ be complete FAs.
Then,
\begin{itemize}
  \item  
    $\autof{\varphi_1 \land \varphi_2} = (Q_{\varphi_1} \times Q_{\varphi_2},\alphx,
    \delta_{\varphi_1 \land \varphi_2},
    I_{\varphi_1} \times I_{\varphi_2},\acc_{\varphi_1 \land \varphi_2})$
    where
    \begin{itemize}
      \item  
        $\delta_{\varphi_1 \land \varphi_2} = \{(q_1, q_2) \ltr \sigma (p_1, p_2) \mid
        q_1 \ltr \sigma p_1 \in \delta_{\varphi_1}, q_2 \ltr \sigma p_2 \in
        \delta_{\varphi_2}\}$ and

      \item
        $\acc_{\varphi_1 \land \varphi_2}((q_1, q_2) \ltr
        \sigma (p_1, p_2)) \defiff \acc_{\varphi_1}(q_1 \ltr \sigma p_1) \land
        \acc_{\varphi_2}(q_2 \ltr \sigma p_2)$.
    \end{itemize}


  \item  
    $\autof{\varphi_1 \lor \varphi_2} = (Q_{\varphi_1} \times Q_{\varphi_2},\alphx,
    \delta_{\varphi_1 \lor \varphi_2},
    I_{\varphi_1} \times I_{\varphi_2},\acc_{\varphi_1 \lor \varphi_2})$
    where
    \begin{itemize}
      \item  
        $\delta_{\varphi_1 \lor \varphi_2} = \{(q_1, q_2) \ltr \sigma (p_1, p_2) \mid
        q_1 \ltr \sigma p_1 \in \delta_{\varphi_1}, q_2 \ltr \sigma p_2 \in
        \delta_{\varphi_2}\}$ and

      \item
        $\acc_{\varphi_1 \lor \varphi_2}((q_1, q_2) \ltr
        \sigma (p_1, p_2)) \defiff \acc_{\varphi_1}(q_1 \ltr \sigma p_1) \lor
        \acc_{\varphi_2}(q_2 \ltr \sigma p_2)$.
    \end{itemize}

  \item
    $\autof{\neg \varphi_1} = (2^{Q_{\varphi_1}}, \alphx,
     \delta_{\neg \varphi_1},\{I_{\varphi_1}\}, \acc_{\neg\varphi_1})$ where
     \begin{itemize}
       \item  $\delta_{\neg \varphi_1} =
         \{S \ltr \sigma T \mid T = \{p \in Q_{\varphi_1} \mid \exists q \in
         S\colon q \ltr \sigma p \in \delta_{\varphi_1}\}$ and
       \item  $\acc_{\neg \varphi_1}(S \ltr \sigma T) \defiff
         \forall q \in S\,\forall p \in T\colon \neg \acc_{\varphi_1}(q \ltr \sigma p)$.

     \end{itemize}

\end{itemize}

Existential quantification is more complicated.
Given a~formula $\formulaof{\exists x(\varphi)}$ and the FA
$\autof{\varphi} = (Q_{\varphi}, \alphx, \delta_\varphi, I_{\varphi},
\acc_\varphi)$, a~word~$w$ should be accepted by $\autof{\exists x(\varphi)}$ iff 
there is a~word~$w'$ accepted by~$\autof \varphi$ s.t.~$w$ and~$w'$
are the same on all tracks except the track for~$x$.
One can perform projection of~$x$ out of~$\autof
\varphi$, i.e., remove the~$x$ track from all its transitions.
This is, however, insufficient.
For instance, consider the model $\{x \mapsto 7, y \mapsto -4\}$, encoded into
the (shortest) word $\twosymb 1 0 \twosymb 1 0 \twosymb 1 1
\twosymb 0 1$ (we use the notation $\twosymb x y$).
When we remove the $x$-track from the word, we obtain $\onesymb 0\!\! \onesymb 0\!\!
\onesymb 1\!\! \onesymb 1$, which encodes the assignment $\{y \mapsto -4\}$.
It is, however, not the shortest encoding of the assignment; the shortest
encoding is~$\onesymb 0\!\! \onesymb 0\!\! \onesymb 1$.
Therefore, we further need to modify the FA obtained after projection to also
accept words that would be accepted if their sign bit were arbitrarily
extended, which we do by reachability analysis on the FA.
Formally,
    $\autof{\exists x(\varphi)} = (Q_{\varphi}, \alphx, \delta_{\exists
    x(\varphi)}, I_{\varphi}, \acc_{\exists x(\varphi)})$ where
    %
    \begin{itemize}
      \item  $\delta_{\exists x(\varphi)} = \{q \ltr {\sigma'} p \mid \exists
        q \ltr \sigma p \in \delta_{\varphi}\colon \sigma' \in \proj_x(\sigma)\}$ and
      \item  $\acc_{\exists x(\varphi)}(q \ltr \sigma p) \defiff
        \hspace*{-3mm}\displaystyle\bigvee_{\sigma'\in\proj_x(\sigma)}\hspace*{-3mm}
    \acc_\varphi(q \ltr{\sigma'} p) \lor
        \exists r,s \in \reach(\{p\}, \proj_x(\sigma)) \colon
        \hspace*{-3mm}\displaystyle\bigvee_{\sigma'\in\proj_x(\sigma)}\hspace*{-3mm}
        \acc_\varphi(r \ltr{\sigma'} s).$
    \end{itemize}

\vspace{-3mm}
After defining the base and inductive cases for constructing the FA~$\autof
\varphi$, we can establish the connection between its language and the models of~$\varphi$.
For a~word $w = a_1\dots a_n
\in\alphx$ and a variable $x \in \vars$, we define $w_x = a_1(x)\dots a_n(x)$,
i.e., $w_x$~extracts the binary number assigned to variable~$x$ in~$w$.
For a~binary assignment $\asgn$ of a LIA formula~$\varphi$, we define its language as 
$\wordof{\asgn} = \{ w \in \alphx^* \mid \forall x \in \vars\colon w_x = \asgn(x) \}$. 
We lift the language to sets of binary assignments as usual.

\begin{theorem}\label{thm:semantics}
  Let $\varphi$ be a~LIA formula. Then $\langof{\autof \varphi} = \langof{\semof{\formulaof \varphi}}$.
\end{theorem}
\begin{proof}
  Follows from~\cite[Lemma~5]{BoudetC96}.
\end{proof}


\newcommand{\figPostInductive}[0]{ 
\begin{wrapfigure}[15]{r}{6.6cm}
\vspace*{-12mm}
\begin{align*}
  \postof{\formulaof{\varphi_1 \land \varphi_2}}{\sigma} ~ &{} \defeq~ \formulaof{\postof{\varphi_1}{\sigma} \land \postof{\varphi_2}{\sigma}}\\[-1mm]
  \postof{\formulaof{\varphi_1 \lor \varphi_2}}{\sigma} ~&{} \defeq~ \formulaof{\postof{\varphi_1}{\sigma} \lor \postof{\varphi_2}{\sigma}}\\[-1mm]
  \postof{\formulaof{\neg \varphi}}{\sigma} ~ &{} \defeq~ \formulaof{\neg\postof{\varphi}{\sigma}}\\[-1mm]
  \postof{\formulaof{\exists x(\varphi)}}{\sigma} ~ &{} \defeq~ \formulaof{\textstyle{\exists x\left(\bigvee_{\sigma'\in\proj_x(\sigma)}\postof{\varphi}{\sigma'}\right)}}\\[1mm]
  \finof{\formulaof{\neg \varphi}}{\sigma} ~ &{} \defiff~ \neg\finof{\formulaof{\varphi}}{\sigma}\\[-1mm]
  \finof{\formulaof{\varphi_1 \land \varphi_2}}{\sigma} ~ &{} \defiff~ \finof{\formulaof{\varphi_1}}{\sigma} \land \finof{\formulaof{\varphi_2}}{\sigma} \\[-1mm]
  \finof{\formulaof{\varphi_1 \lor \varphi_2}}{\sigma} ~ &{} \defiff~ \finof{\formulaof{\varphi_1}}{\sigma} \lor \finof{\formulaof{\varphi_2}}{\sigma} \\[-1mm]
  \finof{\formulaof{\exists x(\varphi)}}{\sigma} ~ &{} \defiff~ 
    \exists \formulaof{\psi} \in
    \reach(\{\formulaof{\varphi}\}, \proj_x(\sigma))
    \colon\\[-1mm]
    &\hspace{6mm}\textstyle\bigvee_{\sigma' \in \proj_x(\sigma)}\hspace{-0mm}\finof{\formulaof{\psi}}{\sigma'}
\end{align*}
\vspace{-9mm}
\caption{$\post$ and $\fin$ for non-atomic formulae.}
\label{fig:post_fin_inductive}
\end{wrapfigure}
}

\newcommand{\figStructSim}[0]{ 
\begin{wrapfigure}[10]{r}{8.3cm}
\vspace*{-10mm}
\hspace*{-2mm}
\scalebox{0.87}{
\begin{minipage}{8cm}
\begin{align*}
  \formulaof{\vecof {a_1} \cdot \vecof{x_1} \leq c_1} \structsim \formulaof{\vecof {a_2} \cdot \vecof{x_2} \leq c_2} ~&\textstyle \defiff~  \vecof{a_1} = \vecof{a_2} \wedge \vecof{x_1} = \vecof{x_2} \wedge c_1 \leq c_2 \\[-1mm]
  \textstyle{\formulaof{\bigwedge_{i = 1}^n \varphi_i} \structsim \formulaof{\bigwedge_{j = 1}^m \psi_j}} ~ &\textstyle \defiff~
  \forall 1 \leq j \leq m ~\exists 1 \leq i \leq n \colon \formulaof{\varphi_i} \structsim \formulaof{\psi_j} \\[-1mm]
  \textstyle{\formulaof{\bigvee_{i = 1}^n \varphi_i} \structsim \formulaof{\bigvee_{j = 1}^m \psi_j}} ~ &\textstyle \defiff~ \forall 1 \leq i \leq n ~\exists 1 \leq j \leq m\colon \formulaof{\varphi_i} \structsim \formulaof{\psi_j} \\[-1mm]
  \formulaof{\neg\varphi} \structsim \formulaof{\neg\psi} ~ &\textstyle \defiff~ \formulaof{\psi} \structsim \formulaof{\varphi} \\[-1mm]
  \formulaof{\exists x(\varphi)} \structsim \formulaof{\exists x(\psi)} ~ &\textstyle \defiff~ \formulaof{\varphi} \structsim \formulaof{\psi}
\end{align*}
\end{minipage}
}
\vspace{-3mm}
\caption{Definition of the subsumption preorder~$\structsim$ (we omit cases implied by reflexivity).}
\label{fig:structsim}
\end{wrapfigure}
}

\newcommand{\figExampleSimple}[0]{ 
\begin{figure}[t]
\begin{center}
\begin{tikzpicture}
[
  automaton,
  node distance=30mm
]

\tikzstyle{rectstate}=[state,rectangle,rounded corners]
\tikzstyle{highlighted_edge} = [teal, thick, dashed];

\node[initial,rectstate] (q) {$\begin{array}{c}x \leq 1000\\ {}\land -x \leq 0 \land {} \\ x \modof{257} 255 \end{array}$};
\node[rectstate,right of=q,yshift=15mm,opacity=0.5] (q0) {$\begin{array}{c}x \leq 500\\ {}\land -x \leq 0 \land {} \\ x \modof{257} 256 \end{array}$};
\node[rectstate,right of=q] (q1) {$\begin{array}{c}x \leq 499\\ {}\land -x \leq 0 \land {} \\ x \modof{257} 127 \end{array}$};
\node[rectstate,right of=q1,opacity=0.5,yshift=15mm] (q10) {$\begin{array}{c}x \leq 249\\ {}\land -x \leq 0 \land {} \\ x \modof{257} 192 \end{array}$};
\node[rectstate,right of=q1,opacity=0.5] (q11) {$\begin{array}{c}x \leq 249\\ {}\land -x \leq 0 \land {} \\ x \modof{257} 63 \end{array}$};

\node[rectstate,teal,fill=white,left of=q0,node distance=20mm] (q0new) {$\begin{array}{c}x = 256\end{array}$};
\node[rectstate,teal,fill=white,right of=q10,node distance=20mm] (q10new) {$\begin{array}{c}x = 192\end{array}$};
\node[rectstate,teal,fill=white,right of=q11,node distance=20mm] (q11new) {$\begin{array}{c}x = 63\end{array}$};

\draw[->] (q) to node[labsymb] {$\onesymb 0$} (q0);
\draw[->] (q) to node[labsymb] {$\onesymb 1$} (q1);

\draw[->] (q1) to node[labsymb] {$\onesymb 0$} (q10);
\draw[->] (q1) to node[labsymb] {$\onesymb 1$} (q11);

\draw[-, highlighted_edge] (q0) edge node[above] {$\Leftrightarrow$} (q0new);
\draw[-, highlighted_edge] (q10) edge node[above] {$\Leftrightarrow$} (q10new);
\draw[-, highlighted_edge] (q11) edge node[above] {$\Leftrightarrow$} (q11new);

\end{tikzpicture}
\end{center}
\vspace{-5mm}
\caption{Example of rewriting formulae in the FA for $\protect\formulaof{0 \leq x
  \leq 1000 \land x \modof{257} 255}$.}
\label{fig:example_simple}
\vspace{-5mm}
\end{figure}
}

\newcommand{\figExamplePruning}[0]{ 
\begin{figure}[t]
\begin{center}
\begin{tikzpicture}
[
  automaton,
  node distance=30mm
]

\tikzstyle{rectstate}=[state,rectangle,rounded corners]
\tikzstyle{highlighted_edge} = [teal, thick, dashed];

\node[initial,rectstate] (q1) {$\begin{array}{c}\exists x(7x \leq 1000)\end{array}$};
\node[rectstate,right of=q1,node distance=35mm,opacity=0.5] (q2) {$\begin{array}{c}\exists x(7x \leq 500 \lor 7x \leq 496) \end{array}$};
\node[rectstate,right of=q2,node distance=45mm,opacity=0.5] (q3) {$\begin{array}{r}\exists x(7x \leq 250 \lor 7x \leq 246~ {} \\ {}\lor 7x \leq 248 \lor 7x \leq 244) \end{array}$};

\node[rectstate,teal,fill=white,below of=q2,node distance=15mm] (q2new) {$\begin{array}{c}\exists x(7x \leq 500)\end{array}$};
\node[rectstate,teal,fill=white,below of=q3,node distance=15mm] (q3new) {$\begin{array}{c}\exists x(7x \leq 250)\end{array}$};

\draw[->] (q1) to node[labsymb] {$\onesymb{}$} (q2);

\draw[->] (q2) to node[labsymb] {$\onesymb{}$} (q3);

\draw[->] (q2new) to node[labsymb] {$\onesymb{}$} (q3new);

\draw[-, highlighted_edge] (q2) edge node[right] {$\Updownarrow$} (q2new);
\draw[-, highlighted_edge] (q3) edge node[right] {$\Updownarrow$} (q3new);

\end{tikzpicture}
\end{center}
\vspace{-6mm}
\caption{Example of disjunction pruning in the FA for $\protect\formulaof{\exists x(7x \leq 1000)}$.}
\label{fig:example_pruning}
\vspace{-4mm}
\end{figure}
}

\vspace{-0.0mm}
\section{Derivative-based Construction for Nested Formulae}\label{sec:opt}
\vspace{-0.0mm}

\figPostInductive  

This section lays down the basics of our approach to interconnecting the
automata with the algebraic approach for quantified LIA.
%
%
We aim at using methods and ideas from the algebraic approach to circumvent the
large intermediate automata constructed along the way before obtaining the
small DFAs (cf.~\cref{fig:gap}). 
To do that, we need a variation of the automata-based decision procedure
that exposes the states of the target automata without the need of generating
the complete state space of the intermediate automata first.
To~achieve this, we generalize the post-image function $\post$ (and the
acceptance condition~$\fin$) from \cref{sec:classical} to general non-atomic formulae using an approach similar to that
of~\cite{FiedorHJLV17,HavlenaHLV19,Traytel15}, which introduced derivatives of
WS1S/WS$k$S formulae.
Computing formula derivatives produces automata states that are at the same time LIA formulae,
and can be manipulated as such using algebraic methods and reasoning about their integer semantics.
We will then use basic Boolean simplification, antiprenexing, and also ideas from Cooper's algorithm and Omega test~\cite{Cooper72,Pugh91} 
to prune and simplify the state-formulae.
These particular techniques will be discussed in \cref{sec:syntactic,sec:disj_pruning,sec:simplifications}.

\begin{figure}[t]
\begin{center}
\begin{tikzpicture}
[
  automaton,
  node distance=30mm
]

\tikzstyle{rectstate}=[state,rectangle,rounded corners]
\tikzstyle{highlighted_edge} = [teal, thick, dashed];

\node[initial,rectstate] (q) {$\begin{array}{c}x \leq 1000\\ {}\land -x \leq 0 \land {} \\ x \modof{257} 255 \end{array}$};
\node[rectstate,right of=q,yshift=15mm,opacity=0.5] (q0) {$\begin{array}{c}x \leq 500\\ {}\land -x \leq 0 \land {} \\ x \modof{257} 256 \end{array}$};
\node[rectstate,right of=q] (q1) {$\begin{array}{c}x \leq 499\\ {}\land -x \leq 0 \land {} \\ x \modof{257} 127 \end{array}$};
\node[rectstate,right of=q1,opacity=0.5,yshift=15mm] (q10) {$\begin{array}{c}x \leq 249\\ {}\land -x \leq 0 \land {} \\ x \modof{257} 192 \end{array}$};
\node[rectstate,right of=q1,opacity=0.5] (q11) {$\begin{array}{c}x \leq 249\\ {}\land -x \leq 0 \land {} \\ x \modof{257} 63 \end{array}$};

\node[rectstate,teal,fill=white,left of=q0,node distance=20mm] (q0new) {$\begin{array}{c}x = 256\end{array}$};
\node[rectstate,teal,fill=white,right of=q10,node distance=20mm] (q10new) {$\begin{array}{c}x = 192\end{array}$};
\node[rectstate,teal,fill=white,right of=q11,node distance=20mm] (q11new) {$\begin{array}{c}x = 63\end{array}$};

\draw[->] (q) to node[labsymb] {$\onesymb 0$} (q0);
\draw[->] (q) to node[labsymb] {$\onesymb 1$} (q1);

\draw[->] (q1) to node[labsymb] {$\onesymb 0$} (q10);
\draw[->] (q1) to node[labsymb] {$\onesymb 1$} (q11);

\draw[-, highlighted_edge] (q0) edge node[above] {$\Leftrightarrow$} (q0new);
\draw[-, highlighted_edge] (q10) edge node[above] {$\Leftrightarrow$} (q10new);
\draw[-, highlighted_edge] (q11) edge node[above] {$\Leftrightarrow$} (q11new);

\end{tikzpicture}
\end{center}
\vspace{-5mm}
\caption{Example of rewriting formulae in the FA for $\protect\formulaof{0 \leq x
  \leq 1000 \land x \modof{257} 255}$.}
\label{fig:example_simple}
\vspace{-5mm}
\end{figure}

\begin{example}
In \cref{fig:example_simple}, we show an intuitive example of rewriting state formulae
  when constructing the FA for $\formulaof{0 \leq x \leq
  1000 \land x \modof{257} 255}$.
  After reading the first symbol \!$\onesymb 0$\!, 
  the obtained state formula
  $\formulaof{0 \leq x \leq 500 \land x \modof{257} 256}$ is satisfied only
  by~one~$x$, which is~$256$.
  We can therefore rewrite the formula into an equivalent formula~$\formulaof{x=256}$.
  Similar rewriting can be applied to the state obtained after reading
  $\onesymb{1}\!\!\onesymb{0}$ and $\onesymb{1}\!\!{\onesymb 1}$.
  The rest of the automaton constructed from the rewritten states
  $\formulaof{x=256}$, $\formulaof{x=192}$, and $\formulaof{x=63}$ is then of
  a~logarithmic size (each state in the rest will have only one successor based
  on the binary encoding of~256, 192, or 63 respectively, while if we did not
  perform the rewriting, the states would have two successors and the size
  would be linear).
\qed
\end{example}



In \cref{fig:post_fin_inductive}, we extend the derivative post-image
function~$\post$ and the acceptance condition~$\fin$ (cf.\
\cref{fig:post_atoms,fig:fin_atoms}) to non-atomic formulae.
The derivatives mimic the automata constructions in \cref{sec:classical}, with
the exception that at every step, the derivative (and therefore also the state
in the constructed FA) is a~LIA formula and can be treated as such.
One notable exception is $\postof{\formulaof{\exists x(\varphi)}}{\sigma}$,
which, since the $\post$ function is deterministic, in addition to the
projection, also mimics determinisation.
One can see the obtained disjunction-structure as a~set of states from the standard
subset construction in automata.
Correctness of the construction is stated in the following.

\begin{lemma}\label{lem:}
  Let $\formulaof \varphi$ be a~LIA formula and let~$\autof{\varphi}$ be the FA
  constructed by the procedure in this section or any combination of
  it and the classical one. Then $\langof{\autof \varphi} =
  \langof{\semof{\formulaof{\varphi}}}$.
\end{lemma}

\begin{proof}
Follows from preservation of languages of the states/formulae.
\qed
\end{proof}

Without optimizations, 
the derivative-based construction would generate a~larger FA than
the one obtained from the classical construction,
which can perform minimization of the intermediate automata. 
The derivative-based construction cannot minimize the intermediate automata since
they are not available; they are in a sense constructed on the fly within the
construction of the automaton for the entire formula. 
%
Our algebraic optimizations mimic some effect of the minimization on the fly, while constructing the automaton, by simplifying the state formulae and detecting
entailment
between them.

In principle, when we construct a state $q$ of~$\autof \psi$ as a~result of~$\post$,
we could test whether some state~$p$ was already constructed such
that~$\formulaof{\varphi_q} \Leftrightarrow \formulaof{\varphi_p}$ and, if so,
we could merge $p$ and $q$ (drop~$q$ and redirect the edges to~$p$).
This would guarantee us to directly obtain the minimal DFA
for~$\formulaof{\psi}$ (no two states would be language-equivalent).

Solving the LIA equivalence queries precisely is, however, as hard as solving the
original problem.
Even when we restrict ourselves to quantifier-free formulae,
the equivalence problem is $\clcoNP$-complete.
Our algebraic optimizations are thus a cheaper and more practical alternative capable of merging at least some equivalent states. 
We discuss the optimizations in detail in \cref{sec:syntactic,sec:disj_pruning,sec:simplifications} and
also give a~comprehensive example of their effect in \cref{sec:example_big}.

\vspace{-2.0mm}
\section{Simple Rewriting Rules}\label{sec:syntactic}
\vspace{-1.0mm}

The simplest rewriting rules are just common simplifications generally applicable in predicate logic.
Despite their simplicity, they are quite powerful, since their use enables to apply the other optimizations (\cref{sec:disj_pruning,sec:simplifications}) more often.
%
\begin{enumerate}
  \item  We apply the propositional laws of \emph{identity} ($\formulaof{\varphi \lor
    \false}$ = $\formulaof{\varphi}$ and $\formulaof{\varphi \land \true}$ =
    $\formulaof{\varphi}$) and \emph{annihilation} ($\formulaof{\varphi \land
    \false}$ = $\formulaof{\false}$ and $\formulaof{\varphi \lor \true}$ =
    $\formulaof{\true}$) to simplify the formulae.

  \item  We use \emph{antiprenexing}~\cite{HavlenaHLVV20,Egly94} (i.e., pushing
    quantifiers as deep as possible using inverses of prenexing rules~\cite[Chapter~5]{RobinsonV01}).
    This is helpful, e.g., after a~range-based quantifier instantiation
    (cf. \cref{sec:range-based}), which yields a~disjunction.
    Since our formula analysis framework (\cref{sec:simplifications}) only works over conjunctions below
    existential quantifiers, we need to first push existential quantifiers
    inside the disjunctions to allow further applications of the heuristics.

  \item Since negation is implemented as automaton complementation,
    we apply \emph{De Morgan's laws} ($\formulaof{\neg (\varphi_1 \land
    \varphi_2)} \Leftrightarrow \formulaof{ (\neg\varphi_1) \lor (\neg\varphi_2)}$ and $
    \formulaof{\neg (\varphi_1 \lor \varphi_2)} \Leftrightarrow \formulaof{ (\neg\varphi_1)
    \land(\neg\varphi_2)}$) to push negation as deep as possible. The motivation
    is that small subformulae are likely to have small corresponding automata.
    As complementation requires the~underlying automaton to be deterministic,
    complementing smaller automata helps to mitigate the exponential blow-up of
    determinization.

\end{enumerate}

\noindent
Moreover, we also employ the following simplifications valid for LIA:

\begin{enumerate}
  \setcounter{enumi}{3}

  \item We apply simple reasoning based on variable bounds to simplify the formula,
    e.g., $\formulaof{x \ge 0 \land x \le 10 \land x \neq 0} \Leftrightarrow \formulaof{x \ge
    1 \land x \le 10}$, and to prune away some parts of the~formula, e.g.,$
    \formulaof{x \ge 3 \land (\varphi \lor (x = 0 \land \psi))} \Leftrightarrow \formulaof{x \ge 3 \land
    \varphi}$.

  \item We employ rewriting rules aimed at accelerating the automata
    construction by  minimizing the number of variables used in a formula,
    and, thus, avoiding constructing complicated transition relations, e.g., $ \formulaof{\exists
    x_1, x_2 (ay + b_1x_1 + b_2x_2 \equiv_K 0)} \Leftrightarrow \formulaof{\exists x (ay +
    bx \equiv_K 0)}$ where $b = \gcdof{b_1, b_2}$, or $\formulaof{\exists x
    (ay + bx = 0)} \Leftrightarrow \formulaof{ay \equiv_{|b|} 0}$. 
    
  \item We detect conflicts by identifying small isomorphic subformulae, i.e.,
    subformulae that have the same abstract syntax tree, except for renaming of
    quantified variables, for example, $\formulaof{\exists x (x > 3
    \land x + z \le 10)\land \neg (\exists y (y > 3 \land y + z \le 10)} \Leftrightarrow
    \formulaof{\bot}$.
    One can see this as a~variant of DAGification used in \mona~\cite{KlarlundMS02}.
\end{enumerate}

\vspace{-2.0mm}
\section{Disjunction Pruning}\label{sec:disj_pruning}
\vspace{-1.0mm}

\figStructSim    
We prune disjunctions
by removing disjunct implied by other disjuncts.
That is, 
if it holds that $\formulaof{\varphi_2 \lor \cdots \lor
\varphi_k} \Rightarrow \formulaof{\varphi_1}$, then $\formulaof{\varphi_1 \lor \varphi_2 \lor
\cdots \lor \varphi_k}$ can be replaced by just $\formulaof{\varphi_2 \lor \cdots \lor \varphi_k}$.
Testing the entailment precisely is hard, so we use a stronger but cheaper relation of \emph{subsumpion}.
%
%
%
Our subsumption is a preorder (a~reflexive and
transitive relation)~$\structsim$ between LIA formulae in \cref{fig:structsim}\footnote{
The subsumption is similar to the one used in efficient decision procedures for
WS1S/WS$k$S~\cite{FiedorHJLV17,HavlenaHLV19} with two important differences:
\begin{inparaenum}[(i)]
  \item  it can look inside atomic formulae and use semantics of states and
  \item  it does not depend on the initial structure of the initial formula.
\end{inparaenum}
Both of these make the subsumption relation larger.
}.
When we encounter the said macro\-state and establish
$\formulaof{\varphi_1} \structsim  \formulaof{\varphi_2 \lor \cdots \lor
\varphi_k}$, we perform the rewriting.
This optimization has effect mainly in formulae of the form
$\formulaof{\exists x(\psi)}$: their~$\post$ contains a~disjunction of
formulae of a~similar structure.

\begin{restatable}{lemma}{lemSubsumption}
  For LIA formulae~$\formulaof{\varphi_1}$ and~$\formulaof{\varphi_2}$, if 
  $\formulaof{\varphi_1} \structsim \formulaof{\varphi_2}$, then
  $\formulaof{\varphi_1} \Rightarrow \formulaof{\varphi_2}$.
\end{restatable}

\begin{example}
In \cref{fig:example_pruning}, we show an example of pruning disjunctions in the FA for the
formula $\formulaof{\exists x(7x \leq 1000)}$.
\qed
\end{example}

\begin{figure}[t]
\begin{center}
\begin{tikzpicture}
[
  automaton,
  node distance=30mm
]

\tikzstyle{rectstate}=[state,rectangle,rounded corners]
\tikzstyle{highlighted_edge} = [teal, thick, dashed];

\node[initial,rectstate] (q1) {$\begin{array}{c}\exists x(7x \leq 1000)\end{array}$};
\node[rectstate,right of=q1,node distance=35mm,opacity=0.5] (q2) {$\begin{array}{c}\exists x(7x \leq 500 \lor 7x \leq 496) \end{array}$};
\node[rectstate,right of=q2,node distance=45mm,opacity=0.5] (q3) {$\begin{array}{r}\exists x(7x \leq 250 \lor 7x \leq 246~ {} \\ {}\lor 7x \leq 248 \lor 7x \leq 244) \end{array}$};

\node[rectstate,teal,fill=white,below of=q2,node distance=15mm] (q2new) {$\begin{array}{c}\exists x(7x \leq 500)\end{array}$};
\node[rectstate,teal,fill=white,below of=q3,node distance=15mm] (q3new) {$\begin{array}{c}\exists x(7x \leq 250)\end{array}$};

\draw[->] (q1) to node[labsymb] {$\onesymb{}$} (q2);

\draw[->] (q2) to node[labsymb] {$\onesymb{}$} (q3);

\draw[->] (q2new) to node[labsymb] {$\onesymb{}$} (q3new);

\draw[-, highlighted_edge] (q2) edge node[right] {$\Updownarrow$} (q2new);
\draw[-, highlighted_edge] (q3) edge node[right] {$\Updownarrow$} (q3new);

\end{tikzpicture}
\end{center}
\vspace{-6mm}
\caption{Example of disjunction pruning in the FA for $\protect\formulaof{\exists x(7x \leq 1000)}$.}
\label{fig:example_pruning}
\vspace{-4mm}
\end{figure}




\vspace{-5.0mm}
\section{Quantifier Instantiation}\label{sec:simplifications}
\vspace{-2.0mm}


The next optimization is an instance of quantifier instantiation, a well known class of algebraic techniques.
We gather information about the formulae with a~focus on the way
a~particular variable, usually a~quantified one, affects the models of the
whole formula.
If one can find ``the best'' value for such a~variable (e.g., a~value such that
using it preserves all models of the formula), then the (quantified) variable
can be substituted with a~concrete value.
For instance, let $\formulaof{\varphi} = \formulaof{\exists y(x - y \leq 33 \land y \leq 12 \land y \modof 7 2)}$.
The variable~$y$ is quantified so we can think about instantiating it (it will not occur in a~model).
The first atom $\formulaof{x-y \leq 33}$ says that we want to pick~$y$ as large as possible (larger~$y$'s have higher chance to satisfy the inequation), but, on the other hand, the second atom $\formulaof{y \leq 12}$ says that~$y$ can be at most~$12$.
The last atom $\formulaof{y \modof 7 2}$ adds an additional constraint on~$y$.
Intuitively, we can see that the best value of~$y$---i.e., the value that
preserves all models of~$\formulaof \varphi$---would be~$9$, allowing to
rewrite~$\formulaof\varphi$ to~$\formulaof{x \leq 42}$.

To define the particular ways of gathering such kind of information in a~uniform way, we
introduce the following \emph{formula analysis} framework that uses the
function~$\flow$ to extract information from formulae.
Consider a~\emph{meet-semilattice} $(\datadom, \meet)$ where $\undefined \in
\datadom$ is the bottom element.
Let $\atom$ be a~function that, given an atomic formula
$\formulaof{\varphi_{\mathit{atom}}}$ and a~variable~$y\in \vars$, outputs an
element of~$\datadom$ that represents the behavior of~$y$
in~$\formulaof{\varphi_{\mathit{atom}}}$ (e.g., bounds on~$y$).
The function~$\flow$ then aggregates the information from atoms into an
information about the behavior of~$y$ in the whole formula using the meet
operator~$\meet$ recursively as follows:
%
%
%
%
\begin{align}
  \flow(\formulaof{\varphi_{\mathit{atom}}}, y, \atom, \meet) = {} & \atom(\formulaof{\varphi_{\mathit{atom}}}, y) \\
  \flow(\formulaof{\varphi_1 \land \varphi_2}, y, \atom, \meet) = {} &
  \flow(\formulaof{\varphi_1}, y, \atom, \meet) \meet
  \flow(\formulaof{\varphi_2}, y, \atom, \meet)
\end{align}
(By default, a~missing case in the pattern matching evaluates to~$\undefined$.)
We note that the framework is defined only for conjunctions of formulae, which
is the structure of subformulae that was usually causing troubles in our
experiments (cf.\ \cref{sec:experiments}).

The optimizations defined later are based on substituting certain variables in a~formula with
concrete values to obtain an equivalent (simpler) formula.
For this, we extend standard substitution as follows.
Let~$\formulaof{\varphi(\vecof x, y)}$ be a~formula with free variables~$\vecof x = (x_1,
\ldots, x_n)$ and~$y \notin \vecof x$.
For $k \in \integers$, substituting~$k$ for~$y$ in~$\formulaof\varphi$ yields the formula
$\formulaof{\varphi\subst y k}$ obtained in the usual way (with all constant expressions
being evaluated).
For $k = \pm \infty$, the resulting formula is obtained for inequalities containing~$y$ as
\begin{equation}
  (\formulaof{\vecof a \cdot \vecof x + a_y \cdot y \leq c})\subst{y}{k} =
    \formulaof{\true}\quad  \text{if } a_y \cdot k = {-\infty}
\end{equation}
and is undefined for all other atomic formulae.

\vspace{-2.0mm}
\subsection{Quantifier Instantiation based on Formula Monotonicity}\label{sec:monotonicity}
\vspace{-1.0mm}

The first optimization based on quantifier instantiation uses the so-called
\emph{monotonicity} of formulae w.r.t.\ some variables (a~similar technique is
used in the Omega test~\cite{Pugh91}).
Consider the following two formulae:
\begin{equation}
\formulaof{\varphi_1} = \formulaof{\exists y(\psi \land 3y - x \geq z)}
  \quad\text{and}\quad
\formulaof{\varphi_2} = \formulaof{\exists y(\psi \land 3y - x \geq z \land 5y \leq 42)}
\end{equation}
where~$\formulaof{\psi}$ does not contain occurrences of~$y$, and~$x$, $z$ are
free variables in~$\formulaof{\varphi_1}, \formulaof{\varphi_2}$.
For~$\formulaof{\varphi_1}$, since~$y$ is existentially quantified, the
inequation~$\formulaof{3y -x \geq z}$ can be always satisfied by picking
an arbitrarily large value for~$y$, so~$\formulaof{\varphi_1}$ can be simplified
to just~$\formulaof{\psi}$.
On~the other hand, for~$\formulaof{\varphi_2}$, we cannot pick an arbitrarily large~$y$
because of the other inequation $\formulaof{5y \leq 42}$.
We can, however, observe, that $\lfloor\frac{42}{5}\rfloor = 8$ is the largest
value that we can substitute for~$y$ to satisfy $\formulaof{5y \leq 42}$.
As a~consequence, since the possible value of~$y$ in $\formulaof{3y-x \geq z}$
is not bounded from above, we can substitute~$y$ by the value~8, i.e.,
$\formulaof{\varphi_2}$ can be simplified to $\formulaof{\psi \land 24 - x \geq z}$.

Formally, let $c \in \integers \cup \{{+\infty}\}$ and $y \in \vars$.
We say that a~formula $\formulaof{\varphi(\vecof{x},y)}$ is \emph{$c$-best from below
w.r.t.~$y$} if 
\begin{inparaenum}[(i)]
  \item $\semof{\formulaof{\varphi\subst{y}{y_1}}} \subseteq
         \semof{\formulaof{\varphi\subst{y}{y_2}}}$ for all $y_1 \leq y_2 \leq
         c$ (for $c \in \integers$) or for all $y_1 \leq y_2$ (for $c =
         {+\infty}$) and
  \item $\semof{\formulaof{\varphi\subst{y}{y'}}} = \emptyset$ for all $y' > c$.
\end{inparaenum}
Intuitively, substituting bigger values for~$y$ (up to~$c$)
in~$\formulaof{\varphi}$ preserves all models obtained by substituting smaller
values, so~$c$ can be seen as the most conservative limit of~$y$ (and for $c=
{+\infty}$, this means that~$y$ does not have an upper bound, so it can be
chosen arbitrarily large for concrete values of other variables).
%
Similarly, $\formulaof{\varphi(\vecof{x},y)}$ is called \emph{$c$-best from
above} (w.r.t.~$y$) for $c \in \integers \cup \{{-\infty}\}$ if 
\begin{inparaenum}[(i)]
  \item $\semof{\formulaof{\varphi\subst{y}{y_1}}} \subseteq
    \semof{\formulaof{\varphi\subst{y}{y_2}}}$ for all $y_1 \geq y_2 \geq
         c$ (for $c \in \integers$) or for all $y_1 \geq y_2$ (for $c =
         {-\infty}$) and
  \item $\semof{\formulaof{\varphi\subst{y}{y'}}} = \emptyset$ for all $y' < c$.
\end{inparaenum}
If a~formula is $c$-best from below or above, we call it \emph{$c$-monotone} (w.r.t.~$y$).
\begin{restatable}{lemma}{lemMonotoneSubst}\label{lem:monotone-subst}
  Let $c \in \integers \cup \{\pm \infty\}$ and
  $\formulaof{\varphi(\vecof{x},y)}$ be a~formula $c$-monotone w.r.t.~$y$ such
  that $\formulaof{\varphi\subst{y}{c}}$ is defined.
  Then the formula $\formulaof{\exists y(\varphi(\vecof{x},y))}$ is equivalent to
  the formula $\formulaof{\varphi\subst{y}{c}}$.
\end{restatable}

Moreover, the following lemma utilizes formula monotonicity to provide
a~tool for simplification of formulae containing a~modulo atom.

\begin{restatable}{lemma}{lemBounded}\label{lem:bounded}
  Let $c\in\mathbb{Z}$ and $\formulaof{\varphi(\vecof{x},y)}$ be a~formula
  $c$-monotone w.r.t.~$y$ for $c\in \integers$.
  Then, the formula $\formulaof{\exists y(\varphi(\vecof{x}, y) \land y \modof{m} k)}$ is
  equivalent to the formula $\formulaof{\varphi\subst{y}{c'}}$ where 
  \begin{inparaenum}[(i)]
    \item $c' = \max\{ \ell \in \integers \mid  \ell \modof{m} k, \ell \leq c \}$ if $\formulaof{\varphi}$ is $c$-best from below, and 
    \item $c' = \min\{ \ell \in \integers \mid  \ell \modof{m} k, \ell \geq c \}$ if $\formulaof{\varphi}$ is $c$-best from above.
  \end{inparaenum}
\end{restatable}

In general, it is, however, expensive to decide whether a~formula is
$c$-monotone and find the tight~$c$.
Therefore, we propose a~cheap approximation working on the structure of LIA
formulae, which uses the formula analysis function~$\flow$ introduced above.
First, we propose the partial function $\decat(\formulaof{\varphi}, y)$ 
(whose result is in $\integers \cup \{+\infty\}$)
estimating the~$c$ for atomic formulae $c$-best from above w.r.t.~$y$: 
\begin{align}
  \decat(\formulaof{a \cdot y \leq c}, y) ={} &
    \!\!\textstyle\left\lfloor\frac{c}{a}\right\rfloor &&
    \text{ if } a > 0
  \\
  \decat(\formulaof{\vecof a \cdot \vecof x \leq c}, x_i) ={} &
    {+\infty} &&
    \text{ if } a_i = 0 \lor \exists j\colon i \neq j \land a_j \neq 0 \land a_i < 0
\end{align}
Intuitively, if~$y$ is in an inequation $\formulaof{a \cdot y \leq c}$ without
any other variable and $a > 0$, then $y$'s value is bounded from above by
$\left\lfloor\frac{c}{a}\right\rfloor$.
On the other hand, if $y=x_i$ is in an inequation $\formulaof{\vecof a \cdot
\vecof x \leq c}$ where $\vecof a$ has at least two nonzero coefficients and
$y$'s coefficient is negative, or~$y$ does not appear in the inequation at all,
then $y$'s value is not bounded (larger values of~$y$ make
it easier to satisfy the inequation).
The value for other cases is undefined.

Similarly, $\incat(\formulaof{\varphi}, y)$
(with the result in $\integers \cup \{-\infty\}$)
estimates the~$c$ for atomic formulae $c$-best from above:
%
%
\begin{align}
  \incat(\formulaof{a \cdot y \leq c}, y) ={} &
    \!\!\textstyle\left\lfloor\frac{c}{a}\right\rfloor &&
    \text{ if } a < 0
  \\
  \incat(\formulaof{\vecof a \cdot \vecof x \leq c}, x_i) ={} &
    {-\infty} &&
    \text{ if } a_i = 0 \lor \exists j\colon i \neq j \land a_j \neq 0 \land a_i > 0
\end{align}

Based on $\decat$ and $\incat$ and using the $\flow$ framework, we define the functions $\dec$ and $\inc$ estimating the $c$ for
general formulae $c$-best from below and above as
\begin{equation}
\dec(\formulaof{\varphi}, y) = \flow(\formulaof{\varphi}, y, \decat, \min),
~~
\inc(\formulaof{\varphi}, y) = \flow(\formulaof{\varphi}, y, \incat, \max).
\end{equation}
%


For a formula $\formulaof\psi = \formulaof{\exists y(\varphi(\vecof{x}, y) \land y
\modof{m} k)}$, the simplification algorithm then determines whether~$\varphi$ is
$c$-monotone for some~$c$, which is done using the $\inc$ and $\dec$ functions.
In particular, if $\dec(\formulaof\varphi, y) = c$ for some $c\in \integers \cup \{\pm
\infty\}$, we have that~$\formulaof\varphi$ is $c$-best from below w.r.t.~$y$ (analogously for
$\inc$).
Then, in the positive case and if $c\in\mathbb{Z}$, we apply 
\cref{lem:bounded} to simplify the formula $\formulaof\psi$.
If~$\formulaof\psi$ is of the simple form $\formulaof{\exists y(\varphi(\vecof{x}, y))}$
where $\formulaof\varphi$ is $c$-monotone w.r.t.~$y$, we can directly use
\cref{lem:monotone-subst} to simplify~$\formulaof\psi$ to $\formulaof{\varphi\subst{y}{c}}$.

\begin{example}
  Consider the formula $\formulaof\psi = \formulaof{\exists y( x-y \leq 1 \land y \leq -1
  \land y \modof{5} 0)}$.
  In order to simplify~$\psi$, 
  we first need to check if the formula $\formulaof{\varphi(x,y)} = \formulaof{x-y \leq 1
  \land y \leq -1}$ is $c$-monotone.
  Using $\dec$, we can deduce that $\dec(\formulaof{x-y\leq 1}, y) =
  \infty$, $\dec(\formulaof{y\leq -1}, y) = -1$, and hence $\dec(\formulaof\varphi, y) =
  -1$ meaning that $\varphi$ is ($-1$)-best from below w.r.t.\ $y$ ($\inc(\formulaof\varphi,
  y)$ is undefined). 
  \cref{lem:bounded} yields that~$\formulaof\psi$ is equivalent to
  $\formulaof{x\leq -4}$ (using $c' = -5$).
  \qed
\end{example}

\vspace{-2.0mm}
\subsection{Range-Based Quantifier Instantiation}\label{sec:range-based}
\vspace{-0.0mm}

Similarly as in Cooper's elimination algoroithm \cite{Cooper72}, we can compute
the \emph{range} of possible values for a~given variable~$y$ and
instantiating~$y$ with all values in the range.
For instance, $\formulaof{\exists y(y \leq 2 \land 2y \geq 3 \land
x + 3y = 42)}$ can be simplified into $\formulaof{x = 36}$.

To obtain the range of possible values of~$y$ in the formula~$\formulaof{\varphi}$,
we use the formula analysis framework with the following function
$\rangeat$ (whose result is an interval of integers) defined for atomic
formulae as follows:
\begin{align}
  \rangeat(\formulaof{a \cdot y \leq c}, y) ={} &
    \Big({-\infty}, \textstyle\left\lfloor\frac{c}{a}\right\rfloor\Big] &&
    \text{ if } a > 0
  \\
  \rangeat(\formulaof{a \cdot y \leq c}, y) ={} &
    \Big[\textstyle\left\lceil\frac{c}{a}\right\rceil, {+\infty} \Big) &&
    \text{ if } a < 0
  \\
  \rangeat(\formulaof{\vecof a \cdot \vecof x \leq c}, x_i) ={} &
    \Big(\textstyle {-\infty}, {+\infty} \Big) &&
    \text{ if } \exists j,k\colon j \neq k \land  a_j \neq 0 \land a_k \neq 0
\end{align}
%
We then employ our formula analysis framework to get the range of $y$
in~$\formulaof{\varphi}$ using the function $\range(\formulaof{\varphi}, y) =
\flow(\formulaof{\varphi}, y, \rangeat, \cap)$. 

\begin{lemma}
  Let $\formulaof{\psi} = \formulaof{\exists y(\varphi(\vecof{x},y))}$ be a~formula such that
  $\range(\formulaof{\varphi}, y) = [a,b]$ with $a,b \in \integers$.
  Then~$\formulaof{\psi}$ is equivalent to the formula $\formulaof{\bigvee_{a \leq c \leq
  b} \varphi\subst{y}{c}}$.
\end{lemma}

\begin{proof}
  It suffices to notice that for all $c \notin [a,b]$ we have 
    $\semof{\formulaof{\varphi(\vecof x, c)}} = \emptyset$.
    \qed
\end{proof}

In our decision procedure, given a~formula $\formulaof{\psi} = \formulaof{\exists
y(\varphi(\vecof{x},y))}$, if $\range(\formulaof{\varphi}, y) = [a,b]$ for
$a,b\in \integers$ and $b-a \leq N$ for a~parameter~$N$ (set by the user), we
simplify~$\formulaof{\psi}$ to $\formulaof{\bigvee_{a \leq c \leq b} \varphi\subst{y}{c}}$.
In our experiments, we set~$N=0$.


\newcommand{\figLinearization}[0]{
\begin{wrapfigure}[14]{r}{4.1cm}
\vspace*{-9mm}
\hspace*{-3mm}
\scalebox{0.9}{
\begin{minipage}{\linewidth}
%
%
%
%
\begin{tikzpicture} 
    \def \miny {-2}
    \def \maxy {2}
    \def \minm {0}
    \def \maxm {5.2}

    \draw[black!15, line width=0.05pt,  step=0.1] (\miny,\minm) grid (\maxy,\maxm);
    \draw[black!25, line width=1.3pt, step=1] (\miny,\minm) grid (\maxy,\maxm);

    \begin{pgfonlayer}{background}
      \draw[draw=blue!20,fill=blue!5] (-1.9,0.1) rectangle (1.7,5);
    \end{pgfonlayer}

    \draw[thick, latex-latex] (\miny-0.3,0)--(\maxy+0.3,0);
    \draw[thick, latex-latex] (0,\minm-0.3)--(0,\maxm+0.3);
    \node at (\maxy+0.2,0.3) {$y$};
    \node at (0.25,\maxm+0.2) {$m$};

    \foreach \y in {\miny,...,-1}
     \node[below,black] at (\y,0) {$\y 0$};
    \foreach \y in {1,2,...,\maxy}
     \node[below,black] at (\y,0) {$\y 0$};
    \foreach \m in {1,2,...,\maxm}
     \node[left,black] at (0,\m) {$\m 0$};
    \node[below left,black] at (0,0) {$0$};

    \draw[red!70!black!80, very thick, domain=-1.9:1.1, smooth,variable=\x]  plot (\x,1.2-\x);
    \draw[red!70!black!80, very thick, domain=-0.1:1.7, smooth,variable=\x]  plot (\x,4.9-\x);


\end{tikzpicture}
\vspace{-8mm}
\end{minipage}
}
\caption{Modulo linearization.}
\label{fig:linearization}
\end{wrapfigure}
}

\vspace{-1.0mm}
\subsection{Modulo Linearization}\label{sec:mod_lin}
\vspace{-0.0mm}

The next optimization is more complex and helps mainly in practical cases
in the benchmarks containing congruences with large moduli.
It does not substitute the value of a~variable by a~constant, but, instead,
substitutes a~congruence with an equation.

\begin{wrapfigure}[14]{r}{4.1cm}
\vspace*{-9mm}
\hspace*{-3mm}
\scalebox{0.9}{
\begin{minipage}{\linewidth}
%
%
%
%
\begin{tikzpicture} 
    \def \miny {-2}
    \def \maxy {2}
    \def \minm {0}
    \def \maxm {5.2}

    \draw[black!15, line width=0.05pt,  step=0.1] (\miny,\minm) grid (\maxy,\maxm);
    \draw[black!25, line width=1.3pt, step=1] (\miny,\minm) grid (\maxy,\maxm);

    \begin{pgfonlayer}{background}
      \draw[draw=blue!20,fill=blue!5] (-1.9,0.1) rectangle (1.7,5);
    \end{pgfonlayer}

    \draw[thick, latex-latex] (\miny-0.3,0)--(\maxy+0.3,0);
    \draw[thick, latex-latex] (0,\minm-0.3)--(0,\maxm+0.3);
    \node at (\maxy+0.2,0.3) {$y$};
    \node at (0.25,\maxm+0.2) {$m$};

    \foreach \y in {\miny,...,-1}
     \node[below,black] at (\y,0) {$\y 0$};
    \foreach \y in {1,2,...,\maxy}
     \node[below,black] at (\y,0) {$\y 0$};
    \foreach \m in {1,2,...,\maxm}
     \node[left,black] at (0,\m) {$\m 0$};
    \node[below left,black] at (0,0) {$0$};

    \draw[red!70!black!80, very thick, domain=-1.9:1.1, smooth,variable=\x]  plot (\x,1.2-\x);
    \draw[red!70!black!80, very thick, domain=-0.1:1.7, smooth,variable=\x]  plot (\x,4.9-\x);


\end{tikzpicture}
\vspace{-8mm}
\end{minipage}
}
\caption{Modulo linearization.}
\label{fig:linearization}
\end{wrapfigure}

Let $\formulaof{\varphi} = \formulaof{\exists y \exists m(\psi \land y + m
\modof{37} 12)}$ such that~$\formulaof{\psi}$ is $17$-best from below w.r.t.~$y$
and $\range(\formulaof{\psi}, m) = [1,50]$.
Since the modulo constraint $\formulaof{y + m \modof{37} 12}$ contains two variables
($y$~and~$m$), we cannot use the optimization from \cref{sec:monotonicity}.
From the modulo constraint and the fact that~$\formulaof \psi$ is 17-best from below
w.r.t.~$y$, we can infer that it is sufficient to consider~$y$ only in the
interval $[-19,17]$ (we obtained $-19$ as $17-37+1$).
The reason is that any other~$y$ can be mapped to a~$y'$ from the same congruence
class (modulo~$37$) that is in the interval and, therefore, gives the same result in the
modulo constraint.
This, together with the other fact (i.e., $\range(\formulaof{\psi}, m) = [1,50]$)
tells us that it is sufficient to only consider the (possibly multiple) linear
relations between~$y$ and~$m$ in the rectangle $[-19,17] \times [1,50]$
(cf.~\cref{fig:linearization}).
The modulo constraint can, therefore, be substituted by the linear relations to obtain the formula
\begin{equation}
\scalebox{0.87}{
\formulaof{\exists y \exists m(\psi \land ((y \geq -19 \land y \leq 11 \land
y+m = 12) \lor (y \geq -1 \land y \leq 17 \land y+m = 49)))}.
  }
\end{equation}

Although the formula seems more complex than the original one, it avoids the
large FA to be
generated for the modulo constraint (a~modulo constraint with~$\modof k$ needs
an~FA with~$k$ states) and, instead, generates the usually much smaller FAs
for the (in)equalities.

The general rewriting rule can be given by the following lemma:

\begin{restatable}{lemma}{lemLinearizationRewrite} \label{lem:LinearizationRewrite}
Let $\formulaof{\psi(\vec{x}, y, m)}$ be a~formula
s.t.~$\range(\formulaof{\psi}, m) = [r, s]$ \mbox{for $r,s \in \integers$,}
let $\formulaof{\varphi} = \formulaof{\exists y \exists m(\psi
\land a_y \cdot y + a_m \cdot m \modof M k)}$, with $a_y \neq 0 \neq a_m$, and
$\alpha = \frac{M}{\gcdof{a_y, M}}$.
%
    If $\formulaof{\psi}$ is $c$-best from below w.r.t.~$y$, then
    $\formulaof{\varphi}$ is equivalent to the formula
\begin{equation}
\formulaof{\exists y \exists m\left(\psi \land \left(\bigvee_{i=0}^{N-1} a_y \cdot y + a_m \cdot m = k + (\ell_1 + i)\cdot \alpha\right)\right)}
\end{equation}
where
\begin{equation}
  \ell_1\!=\!
  \begin{cases}
  \left\lceil
    \frac{a_y \cdot (c-\alpha+1) + a_m \cdot r - k}
    {\alpha}
  \right\rceil
    & \text{for } \frac{a_m}{a_y} \!>\! 0  \\[2mm]
  \left\lceil
    \frac{a_y \cdot (c-\alpha+1) + a_m \cdot s - k}
    {\alpha}
  \right\rceil
    & \text{for } \frac{a_m}{a_y} \!<\! 0
  \end{cases},
  y_1\!=\!\begin{cases}
    \frac{k+ \ell_1 \cdot \alpha -a_m \cdot r}
    {a_y}
    & \text{for } \frac{a_m}{a_y} \!>\! 0  \\[2mm]
    \frac{k+ \ell_1 \cdot \alpha -a_m \cdot s}
    {a_y}
    & \text{for } \frac{a_m}{a_y} \!<\! 0
  \end{cases},
\end{equation}
and
\begin{equation}
  N =
  \left\lceil
    \frac{
      a_y(c - y_1) + a_m(s - r) + 1
    }{\alpha}
  \right\rceil.
\end{equation}

\end{restatable}

%

Due to space constraints, we omit a~similar lemma for the
case when~$\formulaof{\psi}$ is $c$-best from above.
In our implementation, we use the linearization if the~$N$ from
\cref{lem:LinearizationRewrite} is~1, which is sufficient with many practical
cases with large moduli.

\newcommand{\figAutomataBased}[0]{
\begin{figure}[t]
  \centering
\includegraphics[width=6cm,keepaspectratio]{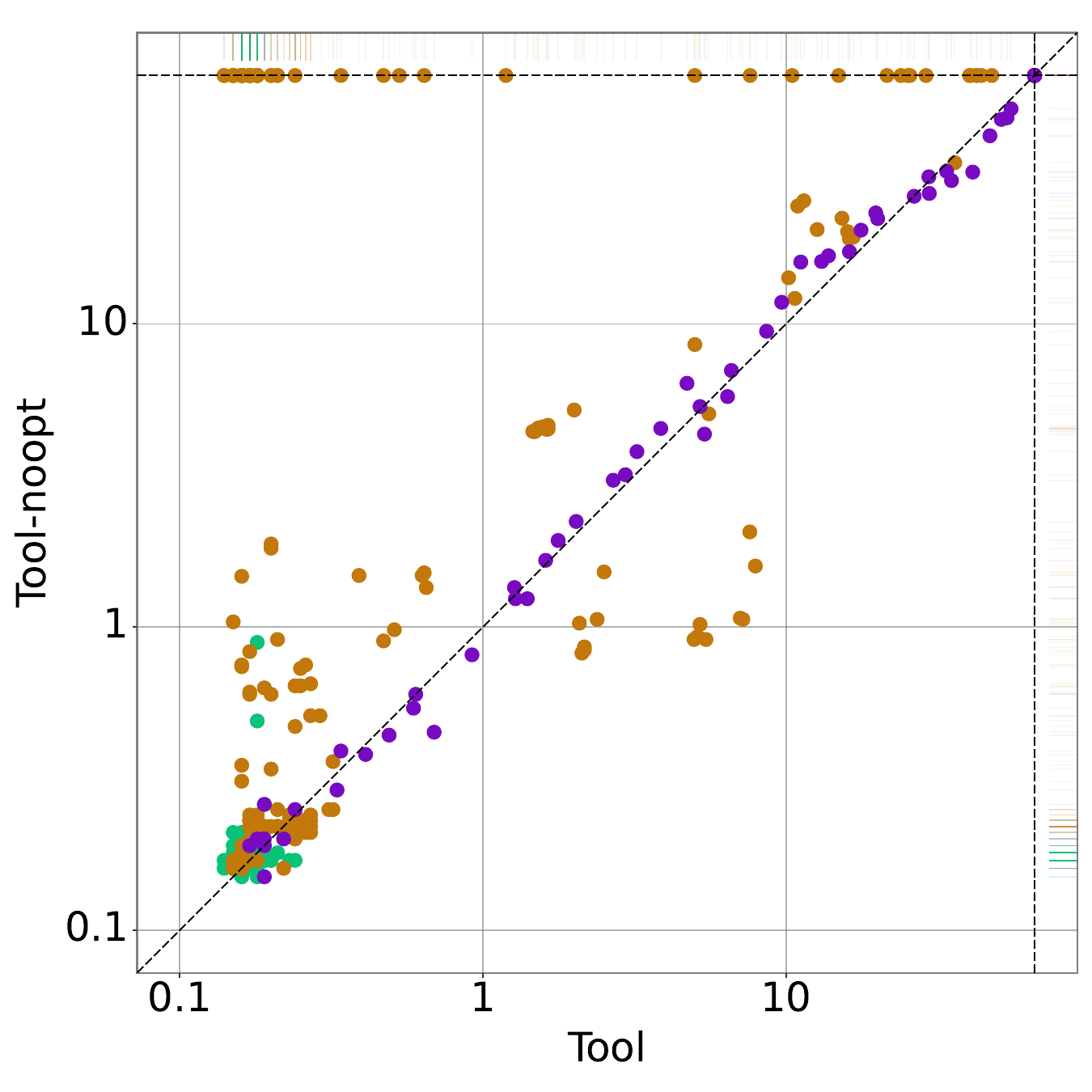}  
\includegraphics[width=6cm,keepaspectratio]{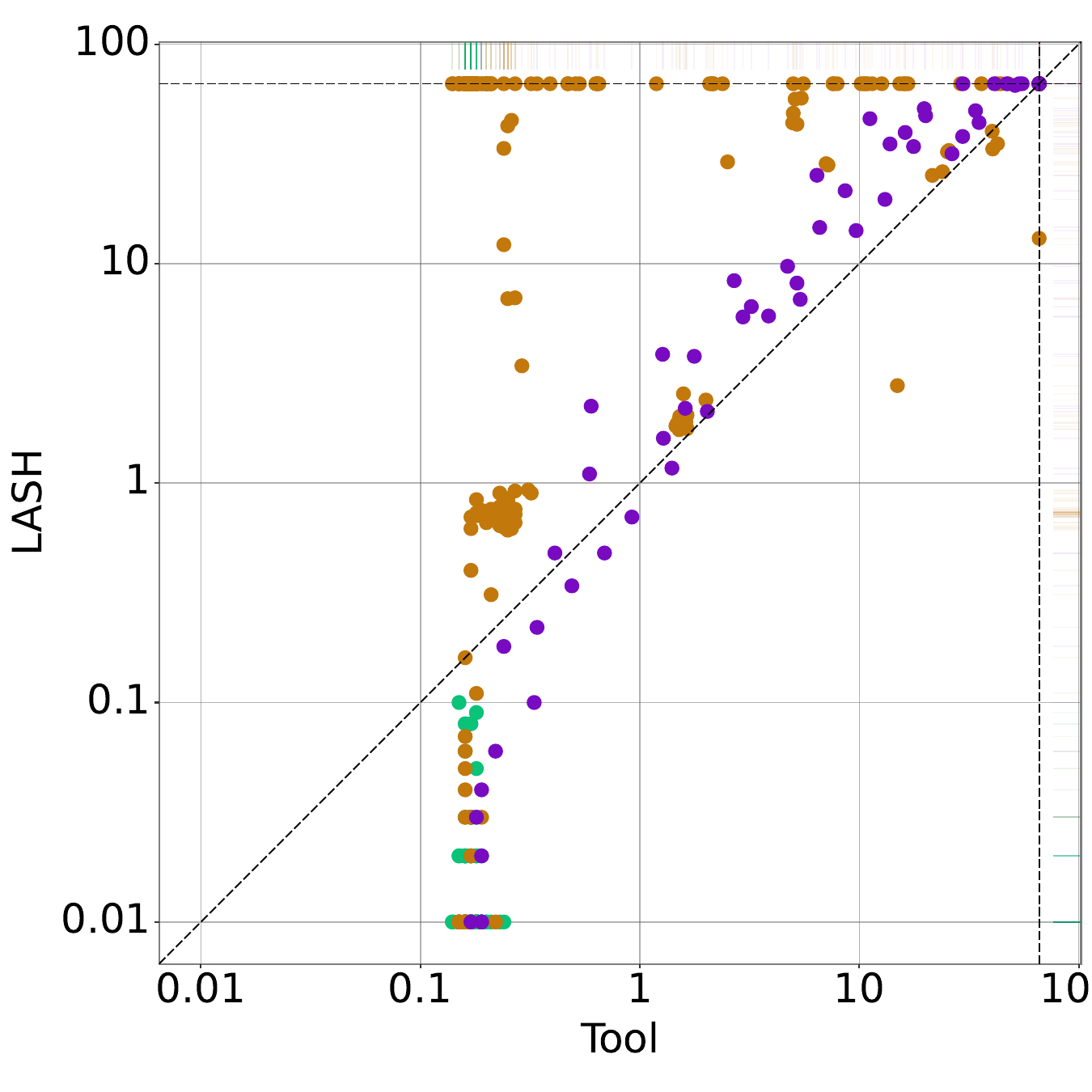}  
\caption{Comparison of automata-based LIA solvers on formulae from SMT-COMP:
  \RGBcircle{0,255,0}\,\texttt{UltimateAutomizer} (153),
  \RGBcircle{186,109,13}\,\texttt{UltimateAutomizerSvcomp2019} (219)
  and \RGBcircle{120,10,194}\,\texttt{Frobenius Coin Problem} (55).
  Times are in seconds, axes are logarithmic.
  Dashed lines represent timeouts (60\,s).}
\label{fig:automata_based}
\vspace{-3mm}
\end{figure}
}

\newcommand{\figFrobenius}[0]{
\begin{figure}[t]
  \centering
\captionsetup{font=footnotesize}
\vspace*{-8mm}
\hspace*{-4mm}
\includegraphics[width=\textwidth,keepaspectratio]{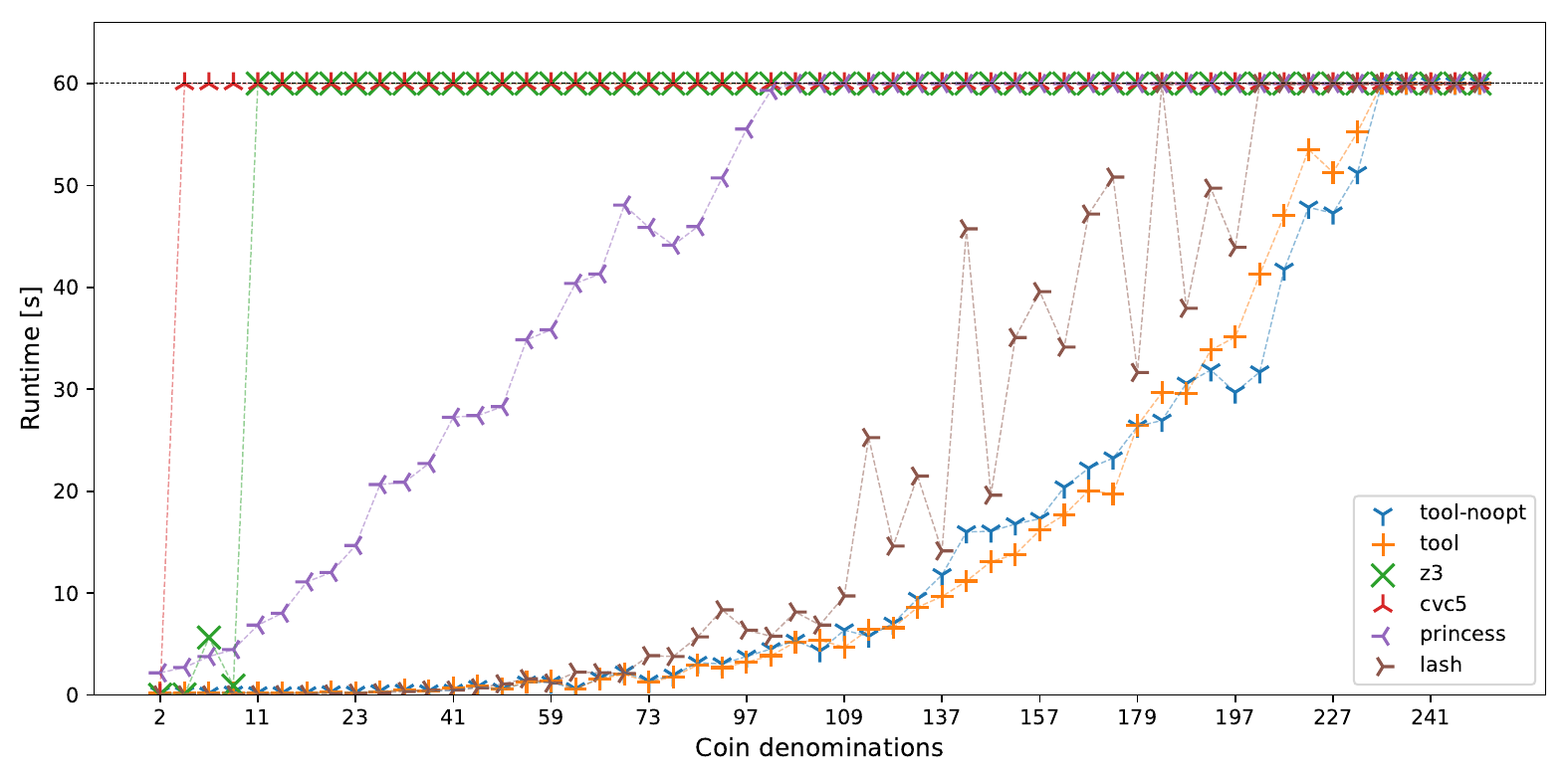}
\vspace{-4mm}
\caption{Run times of solvers on \frobenius.}
\label{fig:frobenius}
\end{figure}
}


\newcommand{\tabResults}[0]{
\begin{table}[t]
\captionsetup{font=footnotesize}
\hspace*{0mm}
\caption{Comparison of solvers on formulae from the \smtcomp and the \frobenius benchmark.  Times
  are given in seconds.  The columns \textbf{wins} and \textbf{losses} show
  numbers of formulae where \amaya performed better and worse
  (wins/losses caused by timeouts are in parentheses).}
\hspace*{-3.5mm}
\scalebox{0.84}{
\begin{tabular}{c}
\newcolumntype{g}{>{\columncolor{Gray!30}}r}
\newcolumntype{f}{>{\columncolor{Gray!30}}l}
\newcolumntype{h}{>{\columncolor{Gray!30}}c}
\begin{tabular}{lgrgrggrrcgrgrggrr}
    \toprule
    & \multicolumn{8}{c}{\smtcomp (372)} & \multicolumn{8}{c}{\frobenius (55)} \\
    \cmidrule(lr){2-9}\cmidrule(lr){11-18}
  \multicolumn{1}{c}{\textbf{solver}} & \multicolumn{1}{h}{\textbf{timeouts}} & \multicolumn{1}{c}{\textbf{mean}}   & \multicolumn{1}{h}{\textbf{median}} & \multicolumn{1}{c}{\textbf{std.\ dev.}}       &   \multicolumn{2}{h}{\textbf{wins}} & \multicolumn{2}{c}{\textbf{losses}}&~~& \multicolumn{1}{h}{\textbf{timeouts}}   & \multicolumn{1}{c}{\textbf{mean}}   & \multicolumn{1}{h}{\textbf{median}} & \multicolumn{1}{c}{\textbf{std.\ dev.}}     &   \multicolumn{2}{h}{\textbf{wins}} & \multicolumn{2}{c}{\textbf{losses}} \\
\cmidrule(lr){1-9}\cmidrule(lr){11-18}

  \cellcolor{GreenYellow}\amaya & \cellcolor{GreenYellow}\textbf{17}& \cellcolor{GreenYellow}1.12   & \cellcolor{GreenYellow}0.26   & \cellcolor{GreenYellow}3.58     &\cellcolor{GreenYellow}    &\cellcolor{GreenYellow}      &\cellcolor{GreenYellow}    &\cellcolor{GreenYellow} &\cellcolor{GreenYellow} & \cellcolor{GreenYellow}\textbf{5} & \cellcolor{GreenYellow}11.79   & \cellcolor{GreenYellow}3.54   & \cellcolor{GreenYellow}16.03    &\cellcolor{GreenYellow}    &\cellcolor{GreenYellow}      &\cellcolor{GreenYellow}    &\cellcolor{GreenYellow} \\
  \amayanoopt                  &  73 & 2.32  & 0.27 & 8.16 &  232 & (56) &  113 & (0) &&  \textbf{5}&  11.54  & 4.06  & 14.65  &  27 & (0)  &  21 & (0) \\
  \lash                        & 114 & 3.04  & \textbf{0.01} & 9.94  & 178 & (98) & 178 & (1) && 9& 15.72  & 5.74  & 20.32  &  37 & (5)  &  14 & (0) \\
\cmidrule(lr){1-9}\cmidrule(lr){10-17}
  \ziii                        &  31& \textbf{0.11}  & \textbf{0.01} & 1.35 &  31 & (28) & 338 & (14)&& 51 &   1.66  &  0.49 &  2.69 &  48 & (46) & 2 & (0) \\
  \cvc                         &  28& 0.20  & 0.02 & 2.42 &  32 & (28) & 340 & (17) && 54 &  \textbf{0.05}  &  \textbf{0.05} &  ---    &  49 & (49) & 1 & (0) \\
  \princess                    &  50& 4.14  & 1.14 & 9.31 &  354 & (40) &  8 & (7) && 13 & 46.32  & 45.92 & 29.03 &  50 & (8) & 0 & (0) \\
\bottomrule
\end{tabular}

\end{tabular}
}
\label{tab:results}
\end{table}
}

\vspace{-1.0mm}
\section{Experimental Evaluation}\label{sec:experiments}
\vspace{-0.0mm}

We implemented the proposed procedure in a~prototype tool called \amaya.
\amaya is written in Python and contains a~basic automata library with
alphabets encoded using \emph{multi-terminal binary decision diagrams}
(MTBDDs), for which it uses the C-based \sylvan library~\cite{DijkP17} (cf.\
\cref{sec:impl_details} for implementation details).
We ran all our experiments on Debian GNU/Linux~12 system with Intel(R) Xeon(R)
CPU E5-2620 v3 @ 2.40\,GHz and 32\,GiB of RAM with the timeout of 60\,s.

\vspace{-1mm}
\paragraph{Tools.}
We selected the following tools for comparison:
\ziii~\cite{MouraB08} (version 4.12.2),
\cvc~\cite{BarbosaBBKLMMMN22} (version 1.0.5),
\princess~\cite{Rummer08} (version 2023-06-19), and
\lash~\cite{WolperB00} (version~0.92).
Out of these, only \lash is an automata-based LIA solver; the other tools are
general purpose SMT solvers with the LIA theory.

\vspace{-1mm}
\paragraph{Benchmarks.}
Our main benchmark comes from SMT-COMP'22~\cite{SMTCOMP22}, in particular, the
categories LIA~\cite{LIA} and NIA (nonlinear integer arithmetic)~\cite{NIA}.
We concentrate on formulae in directories \texttt{UltimateAutomizer} and
(\texttt{20190429-})\texttt{UltimateAutomizer\\Svcomp2019}
of these categories (the main difference between LIA and NIA is
that LIA formulae are not allowed to use the modulo operator) and remove
formulae from NIA that contain multiplication between variables, giving us 372 formulae.
We denote this benchmark as \smtcomp.
The formulae come from verification of real-world C~programs using Ultimate
Automizer~\cite{HeizmannHP13}.
Other benchmarks in the
categories were omitted. Namely, the \texttt{tptp}
benchmark is easy for all tools (every tool finished within 1.3\,s on each formula). 
The \texttt{psyco} benchmark resembles Boolean reasoning more than integer reasoning.
In particular, 
its formulae
contain simple integer constraints (e.g., $x = y + 1$ or just $x = y$) and complex 
Boolean structure with \texttt{ite} operators and  quantified Boolean variables. 
Our prototype tool is not optimized for these features, 
but with a naive implementation of unwinding of \texttt{ite} and with encoding of Boolean variables in a~special automaton track, \amaya could solve 46 out of the 196 formulae in \texttt{psyco}.
%
%

Our second benchmark consists of the \emph{Frobenius coin problem}~\cite{Haase18} asking the following question:
Given a~pair of coins of certain coprime denominations~$a$ and~$b$, what is the
largest number not obtainable as a~sum of values of these coins?
Or, as a~formula,
\begin{equation}\label{eq:frob}
  \varphi(p) \defiff
  (\forall x, y\colon p \neq ax + by)
      \land
  (\forall r (\neg \exists u, v\colon
              r = au + bv) \Rightarrow r \leq p).
\end{equation}
%
Each formula is specified by a pair of denominations $(a,b)$, e.g, $(3,7)$ for which the model is $11$.
Apart from theoretical interest,
the Frobenius coin problem can be used, e.g., for liveness checking of markings 
of conservative weighted circuits (a variant of Petri nets)
\cite{petriNetsFrobenius} or reasoning about automata with
counters~\cite{oopslamatching,HolikLSTVV19,HuW23}.
We created a~family of 55 formulae encoding the problem with various
increasing coin denominations. We denote this benchmark as \frobenius. 
%
The input format of the benchmarks is SMT-LIB~\cite{BarFT-SMTLIB}, which all
tools can handle except \lash---for this, we implemented a~simple translator in
\amaya for translating LIA problems in SMT-LIB into \lash's input format (the
time of translation is not included in the runtime of \lash).

\paragraph{Results.}
We show the results in \cref{tab:results}. For each benchmark we show 
the run time statistics together with the number of timeouts and the number of wins/loses for 
each competitor of \amaya (e.g., the value ``354~(40)'' in the row for
\princess in \smtcomp means that \amaya was faster than \princess on 354
\smtcomp formulae and in 40 cases out of these, this was because \princess
timeouted).
Note that statistics about times tend to be biased in favour of tools that
timeouted more since the timeouts are not counted.

The first part of the table contains automata-based solvers and 
the second part contains general SMT solvers.
We also measure the effect of our optimizations against \amayanoopt,
a~version of the tool that only performs the classical automata-based procedure
from~\cref{sec:classical} without our optimizations.

\begin{table}[t]
\captionsetup{font=footnotesize}
\hspace*{0mm}
\caption{Comparison of solvers on formulae from the \smtcomp and the \frobenius benchmark.  Times
  are given in seconds.  The columns \textbf{wins} and \textbf{losses} show
  numbers of formulae where \amaya performed better and worse
  (wins/losses caused by timeouts are in parentheses).}
\hspace*{-3.5mm}
\scalebox{0.84}{
\begin{tabular}{c}
\newcolumntype{g}{>{\columncolor{Gray!30}}r}
\newcolumntype{f}{>{\columncolor{Gray!30}}l}
\newcolumntype{h}{>{\columncolor{Gray!30}}c}
\begin{tabular}{lgrgrggrrcgrgrggrr}
    \toprule
    & \multicolumn{8}{c}{\smtcomp (372)} & \multicolumn{8}{c}{\frobenius (55)} \\
    \cmidrule(lr){2-9}\cmidrule(lr){11-18}
  \multicolumn{1}{c}{\textbf{solver}} & \multicolumn{1}{h}{\textbf{timeouts}} & \multicolumn{1}{c}{\textbf{mean}}   & \multicolumn{1}{h}{\textbf{median}} & \multicolumn{1}{c}{\textbf{std.\ dev.}}       &   \multicolumn{2}{h}{\textbf{wins}} & \multicolumn{2}{c}{\textbf{losses}}&~~& \multicolumn{1}{h}{\textbf{timeouts}}   & \multicolumn{1}{c}{\textbf{mean}}   & \multicolumn{1}{h}{\textbf{median}} & \multicolumn{1}{c}{\textbf{std.\ dev.}}     &   \multicolumn{2}{h}{\textbf{wins}} & \multicolumn{2}{c}{\textbf{losses}} \\
\cmidrule(lr){1-9}\cmidrule(lr){11-18}

  \cellcolor{GreenYellow}\amaya & \cellcolor{GreenYellow}\textbf{17}& \cellcolor{GreenYellow}1.12   & \cellcolor{GreenYellow}0.26   & \cellcolor{GreenYellow}3.58     &\cellcolor{GreenYellow}    &\cellcolor{GreenYellow}      &\cellcolor{GreenYellow}    &\cellcolor{GreenYellow} &\cellcolor{GreenYellow} & \cellcolor{GreenYellow}\textbf{5} & \cellcolor{GreenYellow}11.79   & \cellcolor{GreenYellow}3.54   & \cellcolor{GreenYellow}16.03    &\cellcolor{GreenYellow}    &\cellcolor{GreenYellow}      &\cellcolor{GreenYellow}    &\cellcolor{GreenYellow} \\
  \amayanoopt                  &  73 & 2.32  & 0.27 & 8.16 &  232 & (56) &  113 & (0) &&  \textbf{5}&  11.54  & 4.06  & 14.65  &  27 & (0)  &  21 & (0) \\
  \lash                        & 114 & 3.04  & \textbf{0.01} & 9.94  & 178 & (98) & 178 & (1) && 9& 15.72  & 5.74  & 20.32  &  37 & (5)  &  14 & (0) \\
\cmidrule(lr){1-9}\cmidrule(lr){10-17}
  \ziii                        &  31& \textbf{0.11}  & \textbf{0.01} & 1.35 &  31 & (28) & 338 & (14)&& 51 &   1.66  &  0.49 &  2.69 &  48 & (46) & 2 & (0) \\
  \cvc                         &  28& 0.20  & 0.02 & 2.42 &  32 & (28) & 340 & (17) && 54 &  \textbf{0.05}  &  \textbf{0.05} &  ---    &  49 & (49) & 1 & (0) \\
  \princess                    &  50& 4.14  & 1.14 & 9.31 &  354 & (40) &  8 & (7) && 13 & 46.32  & 45.92 & 29.03 &  50 & (8) & 0 & (0) \\
\bottomrule
\end{tabular}

\end{tabular}
}
\label{tab:results}
\end{table}

\paragraph{Discussion.}

In the comparison with other SMT solvers, from \cref{tab:results}, 
automata-based approaches are clearly superior to current SMT solvers on \frobenius
(confirming the conjecture made in~\cite{Haase18}). 
\cvc fails already for denominations $(3,5)$ (where the result is~7) and \ziii
follows suite soon; \princess can solve significantly more formulae than \ziii and \cvc, 
but is still clearly dominated by \amaya. 
See \cref{sec:plots} for details.

On the \smtcomp benchmark, 
\amaya can solve the most formulae among all tools.
It has 17 timeouts, followed by \cvc with 28 timeouts (out of 372 formulae).
On individual examples, the comparison of \amaya against \ziii and \cvc almost always falls under one of the two cases:
\begin{inparaenum}[(i)]
  \item  the solver is one or two orders of magnitude faster than \amaya or
  \item  the solver times out.
\end{inparaenum}
This probably corresponds to specific heuristics of \ziii and \cvc taking effect or not, 
while \amaya has a more robust performance, but is still a~prototype and
nowhere near as optimized.
The performance of \princess is, however, usually much worse.  
\amaya is often complementary to the SMT solvers and was able to solve 6
formulae that no SMT solver did.

Comparison with \amayanoopt shows that the optimizations introduced in this paper have a~profound
effect on the number of solved cases (which is
a~proper superset of the cases solved without them). 
This is most visible on the \smtcomp benchmark, where 
\amaya has 56 TOs less than \amayanoopt. 
On the \frobenius benchmark, the results of \amayanoopt and \amaya are 
comparable. Our optimizations had limited impact here since the formulae are
built only from a small number of simple atoms (cf.~\cref{eq:frob}).
In some cases, \amaya takes even longer than \amayanoopt; this is because the lazy
construction explores parts of the state space that would be pruned by the
classical construction (e.g., when doing an intersection with a~minimized FA with an
empty language).
This could be possibly solved by algebraic rules tailored for lightweight unsatisfiability checking.

We also tried to evaluate the effect of individual optimizations by selectively turning them off.
It turns out that the most critical optimizations are the \emph{simple
rewriting rules} (\cref{sec:syntactic}; when turned off, \amaya gave additional
33 timeouts) and \emph{quantifier instantiation} (\cref{sec:simplifications};
when turned off, \amaya gave additional 28 timeouts).
On the other hand, surprisingly, turning off \emph{disjunction pruning}
(\cref{sec:disj_pruning}) did not have a~significant effect on the result.
By itself (without other optimizations), it can help the basic procedure solve
some hard formulae, but its effect is diluted when used with the rest of the
optimizations.
Still, even though it comes with an additional cost, it still has
a~sufficient effect to compensate for this overhead.

Comparing with the older automata-based solver \lash, \amaya 
solves more examples in both benchmarks; \lash 
has 123 TOs in total compared to 22 TOs of \amaya. 
The lower median of \lash on \smtcomp is partially caused by the facts that
\begin{inparaenum}[(i)]
  \item  \lash is a~compiled \C code while \amaya uses a~Python frontend, which
    has a~non-negligible overhead and
  \item  \lash times out on harder formulae.
\end{inparaenum}

\vspace{-0.0mm}
\section{Related Work}\label{sec:label}
\vspace{-0.0mm}

The decidability of Presburger arithmetic was established already at the
beginning of the 20th century by Presburger~\cite{Presburger29} via
\emph{quantifier elimination}.
Over time, more efficient quantifier-elimination-based decision procedures
occurred, such as the one of Cooper~\cite{Cooper72} or the one used within the
Omega test~\cite{Pugh91} (which can be seen as a~variation of Fourier-Motzkin
variable elimination for linear real
arithmetic~\cite[Section~5.4]{KroeningS16}).
The complexity bounds of 2-$\clNEXP$-hardness and 2-$\clEXPSPACE$ membership for
satisfiability checking were obtained by Fischer and Rabin~\cite{FischerR74}
and Berman~\cite{Berman80} respectively.
Quantifier elimination is often considered impractical due to the blow up in
the size of the resulting formula.
\emph{Counterexample-guided quantifier instantiation}~\cite{ReynoldsKK17} is
a~proof-theoretical approach to establish (one-shot) satisfiability of LIA
formulae, which can be seen as a~lazy version of Cooper's
algorithm~\cite{Cooper72}.
It is based on approximating a~quantified formula by a~set of formulae with the
approximation being refined in case it is found too coarse.
The approach focuses on formulae with one alternation, but is also extended to
any number of alternations (according to the authors, the procedure was
implemented in CVC4).

The first \emph{automata-based} decision procedure for Presburger arithmetic
can be obtained from B\"{u}chi's decision procedure for the second-order logic
WS1S~\cite{Buchi60} by noticing that addition is WS1S-definable.
A~similar construction for LIA is used by Wolper and Boigelot
in~\cite{WolperB95}, except that they avoid performing explicit automata product
constructions by using the notion of \emph{concurrent number automata}, which
are essentially tuples of synchronized FAs.

Boudet and Comon~\cite{BoudetC96} propose a~more direct construction of
automata for atomic constraints of the form $a_1 x_1 + \ldots + a_n x_n \sim c$
(for ${\sim} \in \{=, \leq\}$) over natural numbers; we use a~construction
similar to theirs extended to integers (as used, e.g.,
in~\cite{Durand-GasselinH10}).  Moreover, they give a direct construction for
a~conjunction of equations, which can be seen as a~special case of our
construction from \cref{sec:opt}.
Wolper and Boigelot in~\cite{WolperB00} discuss optimizations of the procedure
from~\cite{BoudetC96} (they use the \emph{most-significant bit first} encoding
though), in particular how to remove some states in the construction for
automata for inequations based on subsumption obtained syntactically from the
formula representing the state (a~restricted version of disjunction pruning,
cf.\ \cref{sec:disj_pruning}).
The works~\cite{BoigelotRW98,BoigelotJW01,BoigelotW02,BoigelotJW05} extend
the techniques from~\cite{WolperB00} to solve the mixed \emph{linear integer
real arithmetic} (LIRA) using weak B\"{u}chi automata, implemented in \lash~\cite{LASH}.

\emph{WS1S}~\cite{Buchi60} is a~closely related logic with an automata-based
procedure similar to the one discussed in this paper (as mentioned above,
Presburger arithmetic can be encoded into WS1S).
The automata-based decision procedure for WS1S is, however, of nonelementary
complexity (which is also a~lower bound for the logic), it was, however,
postulated that the sizes of the obtained automata (when reduced or minimized)
describing Presburger-definable sets of integers are bounded by a~tower of
exponentials of a~fixed height.
(3-$\clEXPSPACE$).
This postulate was proven by Klaedtke~\cite{Klaedtke08} (refined later
by Durand-Gasselin and Habermehl~\cite{Durand-GasselinH10} who show that
all automata during the construction do not exceed size 3-$\clEXP$).
The automata-based decision procedure for WS1S itself has been a~subject of
extensive study, making many pioneering contributions in the area of automata
engineering~\cite{GlennG96,HenriksenJJKPRS95,ElgaardKM98,Klarlund99,KlarlundMS02,TopnikWMS06,FiedorHLV15,FiedorHJLV17,HavlenaHLV19,FiedorHLV19,HavlenaHLVV20},
showcasing in the well-known tool
\mona~\cite{HenriksenJJKPRS95,ElgaardKM98,Klarlund99,KlarlundMS02}.

\bibliographystyle{splncs04}
\bibliography{literature}

\appendix
\clearpage

\vspace{-0.0mm}
\section{More Comprehensive Example of Our Optimizations}\label{sec:example_big}
\vspace{-0.0mm}

\newcommand{\figFrameworkIllustr}[0]{
\begin{figure}[t]
  \resizebox{\textwidth}{!}{
    \begin{tikzpicture}[level distance=1.2cm]

\newcommand{\exFormulaState}[4] {$\exists y,m$\LARGE{$\bigwedge$}\hspace*{-10mm} \footnotesize{$\begin{aligned} x + 3m - y &\le #1 \\ y &\le -1 \\ m &\le #2 \\ -m &\le #3 \\ y - m &\equiv_7 #4\end{aligned}$}}
\newcommand{\exTransitionSymbol}[3]
  {$\begin{smallmatrix}m: \\ x: \\ y:\end{smallmatrix} \left[\begin{smallmatrix} \text{\sout{\ensuremath{#1}}} \\ #2 \\ \text{\sout{\ensuremath{#3}}} \end{smallmatrix}\right]$}


  
  \tikzstyle{formula_state} = [draw,
                               rectangle,
                               rounded corners=2];
  \tikzstyle{transition_label} = [scale=1.0, transform shape];
  \tikzstyle{highlighted_edge} = [teal, thick, dashed];

  \node[formula_state, initial, initial text=] (s_9_6_0_0) {\exFormulaState{9}{6}{0}{0}};
  \node[formula_state, right = 1.5cm of s_9_6_0_0]                                      (s_3_2_0_0) {\exFormulaState{3}{2}{0}{0}};
  \node[formula_state, above right = -7mm and 1.5cm of s_3_2_0_0]                      (s_1_1_0_0) {\exFormulaState{1}{1}{0}{0}};
  \node[formula_state, below right = -7mm and 1.5cm of s_3_2_0_0]          (s_0_0_0_0) {\exFormulaState{0}{0}{0}{0}};


  \node[formula_state, right = 0.5cm of s_1_1_0_0, teal, fill=white]           (s_1_1_0_0_equiv) {$x \leq -6$};

  \node[formula_state, above= 10mm of s_3_2_0_0] (s_3_2_0_4) {\exFormulaState{3}{2}{0}{4}};
  \node[formula_state, above = 11mm of s_1_1_0_0]   (s_0_0_0_6) {\exFormulaState{0}{0}{0}{6}};
  \node[formula_state, teal, right = 0.5cm of s_0_0_0_6] (s_0_0_0_6_equiv) {$x \le -1$};

  \node () at ($(s_3_2_0_4)!0.5!(s_3_2_0_0)$) {\huge{$\bigvee$}};
  \node () at ($(s_0_0_0_6)!0.5!(s_1_1_0_0)$) {\huge{$\bigvee$}};
  \node () at ($(s_1_1_0_0)!0.5!(s_0_0_0_0)$) {\huge{$\bigvee$}};

  \node[above=of s_0_0_0_6.south west,xshift=3mm,yshift=-10mm] {$\varphi_1$};
  \node[above=of s_1_1_0_0.south west,xshift=3mm,yshift=-10mm] {$\varphi_2$};
  \node[above=of s_0_0_0_0.south west,xshift=3mm,yshift=-10mm] {$\varphi_3$};

  \draw[->] (s_9_6_0_0) edge node[transition_label, above]        {\exTransitionSymbol{1}{0}{1}} (s_3_2_0_0);
  \draw[->] (s_9_6_0_0) edge node[transition_label, above left]   {\exTransitionSymbol{1}{0}{0}} (s_3_2_0_4);

  \draw[->] (s_3_2_0_0) edge node[transition_label, xshift=0mm,yshift=7mm]  {\exTransitionSymbol{0}{0}{0}} (s_1_1_0_0);
  \draw[->] (s_3_2_0_0) edge node[transition_label, xshift=0mm,yshift=-7mm] {\exTransitionSymbol{1}{0}{1}} (s_0_0_0_0);

  \draw[->] (s_3_2_0_4) edge node[transition_label, above,yshift=2mm]  {\exTransitionSymbol{1}{0}{0}} (s_0_0_0_6);


  \draw[highlighted_edge]  (s_0_0_0_0.60) edge[->] node[right] {$\structsim$} (s_1_1_0_0.300);
  \draw[->, highlighted_edge] (s_1_1_0_0_equiv) edge[] node[right] {$\structsim$} (s_0_0_0_6_equiv);

  \draw[-, highlighted_edge]  (s_0_0_0_6) edge node[above] {$\Leftrightarrow$} (s_0_0_0_6_equiv);
  \draw[-, highlighted_edge]  (s_1_1_0_0) edge node[above] {$\Leftrightarrow$} (s_1_1_0_0_equiv);

  %

  \begin{pgfonlayer}{background}
    \node[draw,dashed,rectangle,fill=black!10,draw=black!70,rounded corners=8pt,inner sep=1.5mm,fit=(s_3_2_0_4) (s_3_2_0_0)] () {};
    \node[draw,dashed,rectangle,fill=black!10,draw=black!70,rounded corners=8pt,inner sep=1.5mm,fit=(s_0_0_0_6) (s_1_1_0_0) (s_0_0_0_0)] () {};
  \end{pgfonlayer}
\end{tikzpicture}
  }
  \caption{Fragment of the generated space for the formula in the example.}
  \label{fig:frameworkIllustr}
\end{figure}
}

\begin{figure}[t]
  \resizebox{\textwidth}{!}{
    \begin{tikzpicture}[level distance=1.2cm]

\newcommand{\exFormulaState}[4] {$\exists y,m$\LARGE{$\bigwedge$}\hspace*{-10mm} \footnotesize{$\begin{aligned} x + 3m - y &\le #1 \\ y &\le -1 \\ m &\le #2 \\ -m &\le #3 \\ y - m &\equiv_7 #4\end{aligned}$}}
\newcommand{\exTransitionSymbol}[3]
  {$\begin{smallmatrix}m: \\ x: \\ y:\end{smallmatrix} \left[\begin{smallmatrix} \text{\sout{\ensuremath{#1}}} \\ #2 \\ \text{\sout{\ensuremath{#3}}} \end{smallmatrix}\right]$}


  
  \tikzstyle{formula_state} = [draw,
                               rectangle,
                               rounded corners=2];
  \tikzstyle{transition_label} = [scale=1.0, transform shape];
  \tikzstyle{highlighted_edge} = [teal, thick, dashed];

  \node[formula_state, initial, initial text=] (s_9_6_0_0) {\exFormulaState{9}{6}{0}{0}};
  \node[formula_state, right = 1.5cm of s_9_6_0_0]                                      (s_3_2_0_0) {\exFormulaState{3}{2}{0}{0}};
  \node[formula_state, above right = -7mm and 1.5cm of s_3_2_0_0]                      (s_1_1_0_0) {\exFormulaState{1}{1}{0}{0}};
  \node[formula_state, below right = -7mm and 1.5cm of s_3_2_0_0]          (s_0_0_0_0) {\exFormulaState{0}{0}{0}{0}};


  \node[formula_state, right = 0.5cm of s_1_1_0_0, teal, fill=white]           (s_1_1_0_0_equiv) {$x \leq -6$};

  \node[formula_state, above= 10mm of s_3_2_0_0] (s_3_2_0_4) {\exFormulaState{3}{2}{0}{4}};
  \node[formula_state, above = 11mm of s_1_1_0_0]   (s_0_0_0_6) {\exFormulaState{0}{0}{0}{6}};
  \node[formula_state, teal, right = 0.5cm of s_0_0_0_6] (s_0_0_0_6_equiv) {$x \le -1$};

  \node () at ($(s_3_2_0_4)!0.5!(s_3_2_0_0)$) {\huge{$\bigvee$}};
  \node () at ($(s_0_0_0_6)!0.5!(s_1_1_0_0)$) {\huge{$\bigvee$}};
  \node () at ($(s_1_1_0_0)!0.5!(s_0_0_0_0)$) {\huge{$\bigvee$}};

  \node[above=of s_0_0_0_6.south west,xshift=3mm,yshift=-10mm] {$\varphi_1$};
  \node[above=of s_1_1_0_0.south west,xshift=3mm,yshift=-10mm] {$\varphi_2$};
  \node[above=of s_0_0_0_0.south west,xshift=3mm,yshift=-10mm] {$\varphi_3$};

  \draw[->] (s_9_6_0_0) edge node[transition_label, above]        {\exTransitionSymbol{1}{0}{1}} (s_3_2_0_0);
  \draw[->] (s_9_6_0_0) edge node[transition_label, above left]   {\exTransitionSymbol{1}{0}{0}} (s_3_2_0_4);

  \draw[->] (s_3_2_0_0) edge node[transition_label, xshift=0mm,yshift=7mm]  {\exTransitionSymbol{0}{0}{0}} (s_1_1_0_0);
  \draw[->] (s_3_2_0_0) edge node[transition_label, xshift=0mm,yshift=-7mm] {\exTransitionSymbol{1}{0}{1}} (s_0_0_0_0);

  \draw[->] (s_3_2_0_4) edge node[transition_label, above,yshift=2mm]  {\exTransitionSymbol{1}{0}{0}} (s_0_0_0_6);


  \draw[highlighted_edge]  (s_0_0_0_0.60) edge[->] node[right] {$\structsim$} (s_1_1_0_0.300);
  \draw[->, highlighted_edge] (s_1_1_0_0_equiv) edge[] node[right] {$\structsim$} (s_0_0_0_6_equiv);

  \draw[-, highlighted_edge]  (s_0_0_0_6) edge node[above] {$\Leftrightarrow$} (s_0_0_0_6_equiv);
  \draw[-, highlighted_edge]  (s_1_1_0_0) edge node[above] {$\Leftrightarrow$} (s_1_1_0_0_equiv);

  %

  \begin{pgfonlayer}{background}
    \node[draw,dashed,rectangle,fill=black!10,draw=black!70,rounded corners=8pt,inner sep=1.5mm,fit=(s_3_2_0_4) (s_3_2_0_0)] () {};
    \node[draw,dashed,rectangle,fill=black!10,draw=black!70,rounded corners=8pt,inner sep=1.5mm,fit=(s_0_0_0_6) (s_1_1_0_0) (s_0_0_0_0)] () {};
  \end{pgfonlayer}
\end{tikzpicture}
  }
  \caption{Fragment of the generated space for the formula in the example.}
  \label{fig:frameworkIllustr}
\end{figure}
Consider the formula
\begin{equation}
\formulaof{\exists y, m(x+3m -y \leq 9 \land y \leq -1 \land m \leq 6 \land -m
  \leq 0 \land y-m \modof 7 0)}
\end{equation}
and see \cref{fig:frameworkIllustr} for a~part of the generated FA (for
  simplicity, we only consider a~fragment of the constructed automaton to
  demonstrate our technique).

Let us focus on the configuration after reading the word $x\colon \!\onesymb 0\!\!\onesymb
0$\!: $\formulaof{\varphi_1 \lor \varphi_2 \lor \varphi_3}$.
First, we examine the relation between $\formulaof{\varphi_2}$ and $\formulaof{\varphi_3}$.
We notice that the two formulae look similar with the only difference being in
two pairs of atoms: $\formulaof{x +3m -y \leq 1}$ and $\formulaof{x+3m-y \leq 0}$,
and $\formulaof{m \leq 1}$ and $\formulaof{m \leq 0}$ respectively.
Since $0 \leq 1$ there are structural subsumptions $\formulaof{x+3m-y \leq 0}
\structsim \formulaof{x+3m-y \leq 1}$ and $\formulaof{m \leq 0}
\structsim \formulaof{m \leq 1}$, which yields $\formulaof{\varphi_3} \structsim
\formulaof{\varphi_2}$, and we can therefore use disjunction pruning
(\cref{sec:disj_pruning}) to simplify $\formulaof{\varphi_1 \lor \varphi_2
\lor \varphi_3}$ to $\formulaof{\varphi_1 \lor \varphi_2}$.

Next, we analyze~$\formulaof{\varphi_1} = \formulaof{\exists y,m(\psi_1)}$.
First, we compute $\range(\formulaof{\psi_1}, m) = [0,0]$ and, based on
that, perform the substitution $\formulaof{\psi_1\subst{m}{0}}$, obtaining
(after simplifications) the formula $\formulaof{\psi_1'} = \formulaof{x-y \leq 0 \land y \leq -1
\land y \modof 7 6}$.
Then, we analyze the behaviour of~$y$ in~$\formulaof{\psi_1'}$ by computing
$\dec(\formulaof{x-y \leq 0 \land y \leq -1}, y) = -1$.
Based on this, we know that $\formulaof{x-y \leq 0 \land y \leq -1}$ is
$(-1)$-best from below, so we can use \cref{lem:bounded} to instantiate~$y$
in~$\formulaof{\psi_1'}$ with~$-5$ (the largest number less than~$-1$ satisfying the
modulo constraint), obtaining the (quantifier-free) formula $\formulaof{x \leq -1}$.

Finally, we focus on~$\formulaof{\varphi_2} = \formulaof{\exists y,m(\psi_2)}$ again.
First, we compute $\range(\formulaof{\psi_2}, m) = [0,1]$ and rewrite
$\formulaof{\varphi_2}$ to $\formulaof{\exists y(\psi_2\subst{m}{0} \lor
\psi_2\subst{m}{1})}$.
After antiprenexing, this will be changed to $\formulaof{\exists
y(\psi_2\subst{m}{0}) \lor \exists y(\psi_2\subst{m}{1})}$.
Using similar reasoning as in the previous paragraph, we can analyze the two
disjuncts in the formula to obtain the formula $\formulaof{x \leq -6 \lor x \leq -8}$.
With disjunction pruning, we obtain the final result of simplification
of~$\formulaof{\varphi_2}$ as the formula $\formulaof{x \leq -6}$.
In the end, again using disjunction pruning, the whole formula
$\formulaof{\varphi_1 \lor \varphi_2 \lor \varphi_3}$ can be simplified to
$\formulaof{x \leq -1}$.

\vspace{-0.0mm}
\section{Implementation Details}\label{sec:impl_details}
\vspace{-0.0mm}

\amaya's is architecturally divided into a frontend written in Python responsible for
preprocessing the input formula and orchestrating the decision procedure, and two
user-selectable backends providing the frontend with automata operations.
The first of \amaya's backends is implemented in Python and is geared towards
experimentation while sacrificing execution speed. The backend is capable of
tracking the origin of automata states all the way to the atomic constraints,
allowing analysis of phenomena in solver's runtime traces, e.g., situations when
applying minimization removing large number of automata states. Having rich
information thus allows one to use impactful minimization as a~natural guide
for finding new formula rewriting rules. The backend represents transitions
symbols explicitly, avoiding cumbersome propagation of state-origin information
that happens when manipulating automata that represent the transition function symbolically.

\amaya's second backend is written in \Cpp and provides algorithms for
automata with their transition function represented symbolically using \emph{multi-terminal
binary decision diagrams} (MTBDDs), for which it uses the \sylvan
library~\cite{DijkP17}. The need for a symbolic representation is natural, as the size of the
automaton alphabet grows exponentially with the number of variables.  
As \sylvan is implemented in C, it would be possible to write a~small wrapper providing
\amaya's frontend with low-level MTBDD manipulation (actually, a~previous
version of \amaya was using such a~wrapper).
However, this approach quickly
becomes impractical, as a nontrivial amount of information needs to be 
exchanged between the frontend and the low-level MTBDD operations. Therefore, \amaya's MTBDD backend
is written entirely in \Cpp and exposes only high-level API calls manipulating entire automata.


\vspace{-0.0mm}
\section{Proofs}\label{sec:proofs}
\vspace{-0.0mm}

\lemSubsumption*

\begin{proof}
By induction.
\begin{itemize}
  \item  Base case: The only interesting base case is $\formulaof{\vecof {a}
    \cdot \vecof{x} \leq c_1} \structsim \formulaof{ \vecof {a} \cdot
    \vecof{x} \leq c_2}$.
    For showing that if $c_1 \leq c_2$, then $\semof{\formulaof{\vecof {a}
    \cdot \vecof{x} \leq c_1}} \subseteq \semof{\formulaof{\vecof {a}
    \cdot \vecof{x} \leq c_2}}$, it suffices to notice that every model of the
    formula on the left-hand side is also a model of the formula on the
    right-hand side.

  \item  Inductive case:  All inductive cases are easy to prove by reasoning
    over the models of the formulae.
  \qed
\end{itemize}
\end{proof}


\lemMonotoneSubst*

\begin{proof}
  The case for $c\in\integers$ follows directly from the definition of monotonicity. 
  For the case of $c= \pm \infty$, we prove only the case of $c = {+\infty}$ and 
  $\formulaof{\varphi}$ being $c$-best from below wrt.~$y$ (the other case can be proven analogously).
  From the definition of best from below formulae, we have that $\semof{\formulaof{\varphi(\vecof{x},y_1)}} \subseteq
  \semof{\formulaof{\varphi(\vecof{x},y_2)}}$ for each $y_1 \leq y_2$. 
  %
  Since $\varphi\subst{y}{\infty}$ is defined, let $A$ be atoms of $\formulaof\varphi$ s.t. the substitution $y/{+\infty}$ yields $\top$ or $\bot$.
  These are all atoms 
  containing with $y$ as a free variable. Further consider the function
  $$
    f(\vecof a \cdot \vecof x + a_y \cdot y \leq c) = \frac{c - \vecof a \cdot \vecof x}{a_y}
  $$
  and based on it the maximum for $\vecof x$ as
  $$
    g(\vecof x) = \max\{ f(\psi)(\vecof x) + 1 \mid \psi \in A \}.
  $$
  Then, we have that $\formulaof{\varphi\subst{y}{\infty}} = \formulaof{\varphi(\vecof{x},g(\vecof{x}))}$ and 
  moreover $\formulaof{\varphi\subst{y}{\infty}} = \formulaof{\varphi(\vecof{x},y)}$ for each $y \geq g(\vecof{x})$. 
  which implies that $\semof{\formulaof{\varphi\subst{y}{\infty}}} \subseteq \semof{\formulaof{\exists y(\varphi(\vecof{x},y))}}$.
  Now consider a model $\asgn$ of the formula $\formulaof{\exists y(\varphi(\vecof{x},y))}$. Then, 
  there is a model $\asgn' = \asgn \cup \{ y \mapsto \ell \}$. Since $\ell < \infty$, from the monotonicity of 
  $\varphi$ and from the reasoning above we get that $\asgn$ is a model of $\formulaof{\varphi\subst{y}{\infty}}$, which concludes the 
  proof.
\qed
\end{proof}


\lemBounded*

\begin{proof}
  We prove the $c$-best-from-below case only. The other case can be proven analogously. 
  First, we show that each model $\asgn$ of $\formulaof{\varphi(\vecof{x},c')}$ is also a model of 
  $\formulaof{\exists y(\varphi(\vecof{x}, y) \land y \modof{m} k)}$. Consider $\asgn' = \asgn \cup \{ y \mapsto c' \}$.
  Since $c' \modof{m} k$, we have that $\asgn'$ is a model of $\formulaof{\varphi(\vecof{x}, y) \land y \modof{m} k}$.

  Now, we prove that each model $\asgn$ of $\formulaof{\exists y(\varphi(\vecof{x}, y) \land y \modof{m} k)}$ is 
  also a model of $\formulaof{\varphi(\vecof{x},c')}$. Since $\asgn$ is a model of the first formula, 
  we have that there is some $\ell$ s.t. $\asgn'= \asgn\cup \{ y \mapsto \ell \}$ is a model of $\formulaof{\varphi(\vecof{x}, y) \land y \modof{m} k}$.
  Since $\formulaof{\varphi}$ is $c$-best from below w.r.t.~$y$, we have that $\ell \leq c$. Secondly, 
  from the second conjunct we get that $\ell \modof{m} k$. From the definition of $c'$ we have that $\ell \leq c'$.
  Finally, if we again use the property of best-from-belowness of $\formulaof{\varphi}$ together with $\ell \leq c'$, 
  we obtain that $\asgn$ is a model of $\formulaof{\varphi(\vecof{x},c')}$.
\qed
\end{proof}



\begin{restatable}{lemma}{lemLinearizationYRestr}\label{lem:linearization_y_restr}
For $a_y, a_m, m, k, c \in \integers$ and $M \in \integers^+$, if there is $y_1
\in \integers$ such that $a_y \cdot y_1 + a_m \cdot m \modof M k$, then there
exist $y_2 \in [c-\alpha+1 , c]$ and $y_3 \in [c, c + \alpha - 1]$ such that
$a_y \cdot y_2 + a_m \cdot m \modof M k$ and
$a_y \cdot y_3 + a_m \cdot m \modof M k$ for $\alpha = \frac{M}{\gcdof{a_y, M}}$.
\end{restatable}

\begin{proof}
Assume that $a_y \cdot y_1 + a_m \cdot m \modof M k$ and let us show how
to choose $y_2$ from the interval such that $a_y \cdot y_2 + a_m \cdot m
\modof M k$.
Rewriting the two modulo constraints into equations, we get the following two linear constraints,
\begin{align}
  a_y \cdot y_1 + a_m \cdot m + d_1 \cdot M &= k \\ 
  a_y \cdot y_2 + a_m \cdot m + d_2 \cdot M &= k,
\end{align}
for some $d_1, d_2 \in \integers$.
From equality of the left-hand sides, we obtain
\begin{equation}
  a_y \cdot y_1 + d_1 \cdot M = a_y \cdot y_2 + d_2 \cdot M
\end{equation}
and, after rearrangement,
\begin{equation}
  a_y \cdot y_1 - a_y \cdot y_2 = d_2 \cdot M - d_1 \cdot M,
\end{equation}
which is equivalent to
\begin{equation}
  a_y (y_1 - y_2) = M (d_2 - d_1).
\end{equation}
Since~$d_2$ and~$d_1$ are both existentially quantified, we can substitute $d_2
  - d_1$ by an existentially quanfitied variable~$d$, obtaining
\begin{equation}
  a_y (y_1 - y_2) = M\cdot d.
\end{equation}
Note that since we are reasoning over integers, we cannot just divide the
right-hand side by~$a_y$.
Instead, let $\beta = \gcdof{a_y, M}$.
Then we can express~$a_y$ and~$M$ as $a_y = a_y' \cdot \beta$ and $M = M'
\cdot \beta$, so
\begin{equation}
  a_y' \cdot \beta \cdot (y_1 - y_2) = M' \cdot \beta \cdot d.
\end{equation}
The $\beta$'s cancel out, so we get
\begin{equation}
  a_y' (y_1 - y_2) = M' \cdot d
\end{equation}
where~$a_y'$ and~$M'$ are coprime.
Because of this, the equation can only hold if~$d$ is a~multiple of~$a_y'$,
i.e., $d = h\cdot a_y'$ for some~$h \in \integers$:
\begin{equation}
  a_y' (y_1 - y_2) = M' \cdot h \cdot a_y'
\end{equation}
Now, $a_y'$'s cancel out, i.e.,
\begin{equation}
  y_1 - y_2 = M' \cdot h\qquad \implies \qquad y_2 = y_1 - M' \cdot h.
\end{equation}
Because $M = M' \cdot \gcdof{a_y, M}$, we have that $M' =
\frac{M}{\gcdof{a_y, M}}$, which is the same as the definition of~$\alpha$.
Therefore, for any interval of the size at least~$\alpha$, there will be
a~value of~$h$ such that $y_2 = y_1 - \alpha \cdot h$ is in the interval.
We can duplicate similar reasoning for~$y_3$.
\qed
\end{proof}


\lemLinearizationRewrite*

\begin{proof}
  From \cref{lem:linearization_y_restr} and the fact that $\formulaof{\psi}$ is
  $c$-decreasing w.r.t.~$y$, it follows that it is sufficient to consider~$y$ in the interval
  $[c-\alpha +1, c]$.
  Therefore, for linearizing the modulo constraint $\formulaof{a_y \cdot y +
  a_m \cdot m \modof M k}$ into a~disjunction of linear equations of the form
  $\formulaof{a_y \cdot y + a_m \cdot m = k + \ell_j \cdot \alpha}$ for different~$\ell_j$ that
  intersect the rectangle $[c-\alpha+1,c] \times [r,s]$, we need to
  \begin{inparaenum}[(a)]
    \item  find the~$\ell_1$ of the left-most line intersecting the rectangle (we
      assume~$y$ is on the horizontal axis and~$m$ is on the vertical axis) and
    \item  find the number~$N$ of lines intersecting the rectangle.
  \end{inparaenum}
  The value of~$\ell_1$ is computed using a~technical excursion into high
  school analytic geometry that we do not reproduce here (we just note that the
  case split is to distinguish lines representing increasing/decreasing
  functions).
  From~$\ell_1$, we compute the~$y$-coordinate of the left-most line
  intersecting the rectangle, denoted as~$y_1$, and also~$N$, both by using
  analytic geometry one more time.
\qed
\end{proof}

\vspace{-2.0mm}
\section{Additional Plots}\label{sec:plots}
\vspace{-1.0mm}

In this section, we present additional plots underlying results presented in \cref{sec:experiments}. 
In particular, in \cref{fig:automata_based} we show scatter plots comparing \amaya with other 
automata-based tools and in \cref{fig:frobenius} we show comparison of \amaya with other 
tools on \frobenius.

\figAutomataBased 

\figFrobenius 

\end{document}